\documentclass[10pt]{article}
\usepackage{geometry}
\usepackage{titling}
\usepackage{titlesec}
\usepackage{setspace}
\usepackage{etoc}

\etocsettocdepth{3}

\titleformat*{\section}{\large\bfseries\sffamily}
\titleformat*{\subsection}{\bfseries\sffamily}
\titleformat*{\subsubsection}{\bfseries\sffamily}
\titleformat*{\paragraph}{\bfseries\sffamily}

\usepackage{parskip}
\pretitle{\begin{flushleft}\Large\bfseries\sffamily}
\posttitle{\par\end{flushleft}\vskip 3ex}
\preauthor{\begin{flushleft}}
\postauthor{\par\end{flushleft}}
\predate{\begin{flushleft}}
\postdate{\par\end{flushleft}\vskip 0.5em}
\date{}

\renewenvironment{abstract}
{\begin{quote}
\vspace{-6ex}
\noindent \par{\bfseries\sffamily \abstractname\hspace{.5em}}}
{\medskip\noindent 
\end{quote}
}

\usepackage{pgfplots}
\usepackage{tikz}
\usepgfplotslibrary{external} 
\usepackage{pgf}
\usepackage{pgfplots}
\usepackage{pgfplotstable}
\usepgfplotslibrary{groupplots} 
\usetikzlibrary{arrows.meta}
\usetikzlibrary{decorations.text}    
\usetikzlibrary{math}
\usetikzlibrary{intersections}

\definecolor{cherry1}{rgb}{0.215686, 0.215686, 0.215686}
\definecolor{cherry2}{rgb}{0.563899, 0.155919, 0.156577}
\definecolor{cherry3}{rgb}{0.747389, 0.178584, 0.180272}
\definecolor{cherry4}{rgb}{0.836168, 0.264453, 0.26819}
\definecolor{cherry5}{rgb}{0.880144, 0.397868, 0.404399}
\definecolor{cherry6}{rgb}{0.911942, 0.567676, 0.576412}

\usepackage{amsfonts,amssymb,amsmath,amsthm}
\usepackage{subcaption}
\usepackage{nicefrac}
\usepackage[inline]{enumitem}
\usepackage{floatrow}
\usepackage{proof-at-the-end}
\usepackage{mathtools}
\usepackage{iftex}
\usepackage{etoolbox}
\usepackage{subcaption}
\usepackage{nicefrac}
\usepackage[inline]{enumitem}
\usepackage{floatrow}
\usepackage{xcolor}

\usepackage[bookmarksopen=true,bookmarks=true,colorlinks=true,breaklinks]{hyperref}
\definecolor{linkcolor}{HTML}{6929C4}
\definecolor{citecolor}{HTML}{0043CE}
\hypersetup{colorlinks={true},linkcolor=linkcolor,citecolor=citecolor}
\usepackage[capitalise]{cleveref}

\pgfkeys{/prAtEnd/global custom defaults/.style={
    text proof={Proof},
    text link={}%{\textit{Proof.} %See the \hyperref[proof:prAtEnd\pratendcountercurrent]{\textit{full proof}} in page~\pageref{proof:prAtEnd\pratendcountercurrent}.}
  }
}

\crefname{assumption}{Assumption}{Assumptions}
\crefname{condition}{Condition}{Conditions}
\crefname{framedtheorem}{Theorem}{Theorems}
\crefname{framedproposition}{Proposition}{Propositions}
\crefname{framedlemma}{Lemma}{Lemmas}

\theoremstyle{plain}

\newtheorem{theorem}{Theorem}[section]
\newtheorem{corollary}{Corollary}[section]
\newtheorem{lemma}[theorem]{Lemma}
\theoremstyle{definition}

\newtheorem{assumption}[theorem]{Assumption}

\theoremstyle{remark}

\usepackage[backend=bibtex,style=alphabetic,minalphanames=3,url=false,doi=false,isbn=false,natbib=true,backref=true,maxbibnames=99]{biblatex}
\DefineBibliographyStrings{english}{%
  backrefpage = {page},% originally "cited on page"
  backrefpages = {pages},% originally "cited on pages"
}

\usepackage{parskip}
\newcommand{\norm}[1]{{\left\lVert #1 \right\rVert}}

\newcommand{\inner}[2]{\left\langle #1, #2 \right\rangle}

\usepackage{tikz}
\usepackage{pgf}
\usepackage{pgfplots}
\usepackage{pgfplotstable}
\usetikzlibrary{spy,calc,dsp,chains}
\pgfplotsset{compat=1.17}
\usepgfplotslibrary{external} 
\allowdisplaybreaks

\title{Royal Statistics Society Template}
\begin{document}

\title{\vspace{-2cm}Analysis of Kinetic Langevin Monte Carlo Under \\
The Stochastic Exponential Euler Discretization\\
From Underdamped All The Way to Overdamped}
\author{
  {\sf{}Kyurae Kim} \\
  \textsf{E-mail}: \href{mailto:kyrkim@engineering.upenn.edu}{kyrkim@engineering.upenn.edu} \\
  {\itshape%
    University of Pennsylvania%
  } \\ \vspace{1.ex}
  {\sf{}Samuel Gruffaz} \\
  \textsf{E-mail}: \href{mailto:samuel.gruffaz@ens-paris-saclay.fr}{samuel.gruffaz@ens-paris-saclay.fr} \\
   {\itshape%
     Centre Borelli, ENS Paris-Saclay
  , France} \\ \vspace{1.ex} 
  {\sf{}Ji Won Park} \\
  \textsf{E-mail}: \href{mailto:park.ji_won@gene.com}{park.ji\_won@gene.com} \\
  {\itshape%
    Prescient Design, Genentech\\
  } \vspace{1.ex} 
  % %
  %
  {\sf{}Alain Oliviero Durmus} \\
  \textsf{E-mail}: \href{mailto:alain.durmus@polytechnique.edu}{alain.durmus@polytechnique.edu} \\
  {\itshape%
    CMAP, CNRS, \'Ecole Polytechnique %
  } 
}

\maketitle

\begin{abstract}
Simulating the kinetic Langevin dynamics is a popular approach for sampling from distributions, where only their unnormalized densities are available.
Various discretizations of the kinetic Langevin dynamics have been considered, where the resulting algorithm is collectively referred to as the kinetic Langevin Monte Carlo (KLMC) or underdamped Langevin Monte Carlo.
Specifically, the stochastic exponential Euler discretization, or exponential integrator for short, has previously been studied under strongly log-concave and log-Lipschitz smooth potentials via the synchronous Wasserstein coupling strategy.
Existing analyses, however, impose restrictions on the parameters that do not explain the behavior of KLMC under various choices of parameters.
In particular, all known results fail to hold in the overdamped regime, suggesting that the exponential integrator degenerates in the overdamped limit.
In this work, we revisit the synchronous Wasserstein coupling analysis of KLMC with the exponential integrator.
Our refined analysis results in Wasserstein contractions and bounds on the asymptotic bias that hold under weaker restrictions on the parameters, which assert that the exponential integrator is capable of stably simulating the kinetic Langevin dynamics in the overdamped regime, as long as proper time acceleration is applied.
\end{abstract}

{\hypersetup{linkbordercolor=black,linkcolor=black}
\tableofcontents
}

\section{Introduction}
Consider a differentiable potential function $U : \mathbb{R}^d \mapsto \mathbb{R}$.
This work focuses on the kinetic Langevin dynamics ${(Z_t = (X_t, V_t))}_{t \geq 0}$ described by the system of equations, for each $t \geq 0$,
\begin{equation}
\begin{split}
    \mathrm{d}X_{t} &= V_t \mathrm{d}t \\
    \mathrm{d}V_{t} &= -\eta \nabla U\left(X_t\right) \mathrm{d}t - \gamma V_t \mathrm{d}t  + \sqrt{2 \gamma \eta } \, \mathrm{d}B_t,
    \label{eq:kinetic_langevin_dynamics}
\end{split}
\end{equation}
where $(B_t)_{t\geq 0}$ is a standard $d$-dimensional brownian motion on the filtered probability space $(\Omega,\mathcal{F},\mathbb{P},(\mathcal{F}_t)_{t\geq 0})$ satisfying typical conditions, $\eta > 0$ is the inverse mass, $\gamma > 0$ is the friction coefficient, and $X_t$ and $V_t$ are respectively referred to as the position and the momentum.
Numerically simulating \cref{eq:kinetic_langevin_dynamics} has been an important application in molecular dynamics for modeling interacting particles~\citep[Eqs. (6.30) (6.31)]{leimkuhler_molecular_2015}.

In recent years, \cref{eq:kinetic_langevin_dynamics} has received massive interest from computational statistics and machine learning for the purpose of drawing samples from distributions with only unnormalized densities.
Specifically, under mild assumptions, the process ${(Z_t)}_{t \geq 0}$ converges to its unique stationary measure~\citep[Prop. 6.1]{pavliotis_stochastic_2014} on $\mathbb{R}^{2 d}$,
\begin{align}
    \pi\left(\mathrm{d}x, \mathrm{d}v\right) \triangleq \frac{1}{Z} \exp\left( - U\left(x\right) \right) \mathrm{d}x \; \mathrm{N}\left(\mathrm{d}v; 0_d, \eta\mathrm{I}_d\right), \; 
    \;\;\text{where}\;\;
    Z \triangleq {\textstyle\int}_{\mathbb{R}^d} \exp\left(-U\left(x\right)\right) \mathrm{d}x \; ,
    \label{eq:stationary_measure}
\end{align}
and $\mathrm{N}(\cdot; 0_d, \eta\mathrm{I}_d)$ is a $d$-dimensional multivariate Gaussian distribution with mean $0_d$ and covariance $\eta \mathrm{I}_d$.
Therefore, even if we have only access to $\nabla U$, \cref{eq:kinetic_langevin_dynamics} can be used to produce samples from $\pi$ in \cref{eq:stationary_measure} as long as \cref{eq:kinetic_langevin_dynamics} is accurately simulated.
This scheme, known as kinetic Langevin Monte Carlo (KLMC), or underdamped Langevin Monte Carlo, has been used in various contexts from generative modeling~\citep{saremi_multimeasurement_2021,dockhorn_scorebased_2021}, producing unbiased estimators of expectations~\citep{chada_unbiased_2024}, multi-armed bandits~\citep{zheng_accelerating_2024}, simulating path measures~\citep{kim_tuning_2025,doucet_scorebased_2022,geffner_langevin_2023,blessing_underdamped_2025} in sequential Monte Carlo~\citep{delmoral_sequential_2006} and annealed importance sampling~\citep{neal_annealed_2001}, marginal likelihood maximization~\citep{oliva_kinetic_2024}, and many more.

For numerical simulation, however, a discretization scheme has to be involved.
The Markov chain resulting from discretization ${(Z_k = (X_k, V_k))}_{k \geq 0}$ is often ``asymptotically'' biased in the sense that, for the discretization step size $h \geq 0$, the stationary distribution of ${(Z_k)}_{k \geq 0}$, $\pi_h$, will be different from $\pi$.
Furthermore, different discretization schemes can behave differently in terms of speed of convergence to stationarity, computational efficiency, dependence on the properties of $U$, and the amount of asymptotic bias.
As such, for the purpose of sampling from \cref{eq:stationary_measure}, various discretizations of \cref{eq:kinetic_langevin_dynamics} have been proposed~\citep{cheng_underdamped_2018,leimkuhler_rational_2013,sanz-serna_wasserstein_2021,foster_shifted_2021,shen_randomized_2019,liu_double_2023} and analyzed~\citep{camrud_second_2024,monmarche_nonasymptotic_2024,ma_there_2021,altschuler_shifted_2025,dalalyan_sampling_2020,durmus_uniform_2025,schuh_convergence_2024,zhang_improved_2023,johnston_kinetic_2024,fu_meanfield_2024,leimkuhler_contraction_2024,monmarche_highdimensional_2021,gouraud_hmc_2025,dalalyan_bounding_2022}.

Among various discretization schemes, we focus on the stochastic exponential Euler scheme, which we will hereafter refer to as the exponential integrator.
It was one of the first discretization schemes to yield a non-asymptotic mixing time guarantee for the kinetic Langevin dynamics~\citep{cheng_underdamped_2018} under the assumptions that $U$ is $\alpha$-strongly convex and $\beta$-Lipschitz smooth through a synchronous Wasserstein coupling analysis~\citep[\S 4.1]{chewi_logconcave_2026}.
The analysis by~\citeauthor{cheng_underdamped_2018} has since been refined \citep{dalalyan_sampling_2020,sanz-serna_wasserstein_2021,leimkuhler_contraction_2024}, which we will formally introduce in the following paragraph.
The exponential integrator has also been studied under the assumption that $U$ is non-strongly convex~\citep{dalalyan_bounding_2022} and the framework of log-Sobolev and Poincar\'e functional inequalities~\citep{ma_there_2021,zhang_improved_2023}.
This work, however, focuses on the strongly convex setting.

To state formal results, consider the KLMC algorithm obtained via the exponential integrator with the initial state $(X_0, V_0) \sim \mu_0$ from some $\mu_0$.
Let us denote the Wasserstein-2 distance as $\mathrm{W}_2(\cdot,\cdot)$.
Also, denote the marginal distribution of KLMC at step  $k \geq 0$ as $\mu_k$ such that $(X_k , V_k) \sim \mu_k$.
Our goal is to guarantee that $\mathrm{W}_2 (\mu_k , \pi) \leq \alpha^{-1/2} \epsilon $ for $\epsilon > 0$.
For this, it has been shown~\citep{sanz-serna_wasserstein_2021,dalalyan_sampling_2020} that at most $\mathrm{O}(\kappa^{3/2} d^{1/2} \epsilon^{-1} \log \epsilon^{-1} )$ steps are sufficient, where $\kappa = \beta/\alpha$ is the condition number of $U$.
This number of steps has also been recently shown to be sufficient for bounding the Kullback-Leibler divergence~\citep{kullback_information_1951} between $\mu_k$ and $\pi$~\citep{altschuler_shifted_2025} under the same conditions.

Unfortunately, the analysis by~\citet{sanz-serna_wasserstein_2021} does not fully shed light on the effect of the parameters $\gamma, \eta$, and the integration step size $h$.
That is, they only obtain a Wasserstein contraction for $\eta \in (0, 4/(\alpha + \beta) ), \gamma = 2$ and when $h$ is smaller than some unknown threshold~\citep[Ex. 4.13]{sanz-serna_wasserstein_2021}.
While the remaining analyses \citep{dalalyan_sampling_2020,altschuler_shifted_2025} provide more general and concrete conditions on $\gamma$ and $h$, they also have limitations.
A well-known aspect of \cref{eq:kinetic_langevin_dynamics} is that, after rescaling the time as $t^{\prime} = \gamma t$, setting $\eta = 1$, and taking the overdamped limit $\gamma \to \infty$, we obtain the overdamped Langevin dynamics~\citep[\S 6.5]{pavliotis_stochastic_2014}
\begin{equation}
    \mathrm{d}X_{t^{\prime}} = -\nabla U\left(X_{t^{\prime}}\right) \mathrm{d}t^{\prime} + \sqrt{2} \, \mathrm{d}B_{t^{\prime}} \; .
    \label{eq:overdamped_langevin}
\end{equation}
The results by \citet[Thm. 2]{dalalyan_sampling_2020} and \citet[Thm. 4.1]{altschuler_shifted_2025} suggest that the exponential integrator is unable to simulate \cref{eq:kinetic_langevin_dynamics} in this regime.
Specifically, they require that the discretization step size $h$ satisfies $h = \mathrm{O}(1/\gamma)$, which means that, as $\gamma \to \infty$, the step size has to degenerate such that $h \to 0$.
A similar conclusion is drawn by \citet{leimkuhler_contraction_2024}, who conclude that the discretized dynamics require $h = \mathrm{O}(1/\gamma)$ to form a Wasserstein contraction.
Meanwhile, they demonstrate that other discretizations, such as the ``OBABO'' discretization~\citep{leimkuhler_rational_2013}, do not necessarily degenerate in the large friction limit.
Now, it cannot be the case that only the exponential integrator cannot attain this limit since, for any $h_{\mathrm{LMC}} > 0$, by setting $h = h_{\mathrm{LMC}} \gamma$, $\eta = 1$, and taking the limit $\gamma \to \infty$, the update rule for the KLMC with the exponential integrator exactly coincides with the Euler--Maruyama discretization of \cref{eq:overdamped_langevin}, widely known as Langevin Monte Carlo (LMC;~\citealp{rossky_brownian_1978,parisi_correlation_1981,grenander_representations_1994}).
This suggests that there is still room for improvement for the synchronous Wasserstein coupling approach.

Furthermore, existing analyses of the asymptotic bias of KLMC with the exponential integrator become vacuous in the overdamped limit.
More concretely, the asymptotic bias of KLMC with the exponential integrator in Wasserstein-2 distance $\mathrm{W}_2\left(\pi_h, \pi\right)$ is well-known to scale as $\mathrm{O}(h)$~\citep[Thm. 9]{cheng_underdamped_2018}.
In the overdamped limit, however, as $\gamma \to \infty$, the time-accelerated step size diverges $h = h_{\mathrm{LMC}} \gamma \to \infty$, making these bounds vacuous.
Instead, one would expect a non-vacuous phase transition into a $\mathrm{O}(h_{\mathrm{LMC}}^{1/2})$ scaling, which is the asymptotic bias of LMC~\citep[Cor. 7]{durmus_highdimensional_2019}.
Previously, no analysis has been able to identify this phase transition.

Given the amount of attention devoted to understanding the exponential integrator, the fact that a theoretical gap remains warrants attention.
Furthermore, compared to other discretizations, the exponential integrator is special in two practical aspects: First, its per-iteration computational cost is lower, requiring only a single evaluation of the gradient $\nabla U$; alternatives such as the OBABO discretization often require two evaluations per step or more.
Second, the transition kernel associated with the exponential integrator is precisely a conditional Gaussian over the original state space $\mathbb{R}^d \times \mathbb{R}^d$.
This makes it particularly relevant for applications that require transition kernels that are conjugate to some other distribution~\citep{heng_controlled_2020,bernton_schrodinger_2019}.

In this work, we refine the synchronous Wasserstein coupling analysis of KLMC with the exponential integrator when $U$ is $\alpha$-strongly convex and $\beta$-Lipschitz smooth.
Additional details on the setup and the assumptions are stated in \cref{section:preliminaries}.
Our contributions are as follows:
\begin{itemize}
    \item \cref{section:convergence_stationarity}: We relax the assumptions on the parameters $h, \gamma, \eta$ required to ensure that the discretized process converges to its (biased) stationary distribution.
    Under a general condition on $h$, $\eta$, and $\gamma$, \cref{thm:general_contraction} establishes a Wasserstein contraction with a rate that depends on these parameters.
    \cref{thm:special_case_convergence_bound} states that this result implies a rate of convergence to stationarity of $\mathrm{O}(h \eta \alpha / \gamma)$.
    In the underdamped regime, this matches previously known rates~\citep[Thm. 6.1]{leimkuhler_contraction_2024}.
    However, our result imposes weaker restrictions on the step size $h$, and can be satisfied even in the overdamped regime.
    Indeed, in the overdamped limit, the contraction rate of KLMC coincides with the corresponding contraction rate of the overdamped Langevin discretized with the Euler--Maruyama scheme (\cref{thm:convergence_overdamped}).
    
    \item \cref{section:asymptotic_bias}: We provide a more general result on the asymptotic bias of the stationary distribution of the discretized process in Wasserstein-2 distance.
    Under conditions sufficient to ensure convergence to stationarity, \cref{thm:asymptotic_bias} provides a bound on the asymptotic bias in Wasserstein-2 distance.
    % The result decomposes the asymptotic bias by the contribution of the discretization error of the position $X_t$ and the momentum $V_t$, which sheds light on the behavior of the KLMC in different regimes.
    Specifically, it shows that the asymptotic bias scales as $\mathrm{O}(h^2 \gamma + h)$ and $\mathrm{O}(h^{1/2} \gamma^{-1/2} + h^{-1/2} \gamma^{-3/2}) $ in the underdamped and overdamped regimes, respectively.
    For $h = h_{\mathrm{LMC}}\gamma$, the bound on the overdamped regime remains non-vacuous even in the overdamped limit $\gamma \to \infty$.
    In fact, the asymptotic bias in the overdamped limit precisely matches known results for LMC obtained via synchronous Wasserstein coupling.
    Furthermore, numerically, the phase transition from underdamped to overdamped appears to happen around the point of $h \gamma = 1.69$.
\end{itemize}

In \cref{section:complexity}, we combine the convergence and the asymptotic bias analyses into a mixing time complexity guarantee.
As a result, we obtain a sampling complexity guarantee that $\mathrm{O}(\kappa^{3/2} d^{1/2} \epsilon^{-1} \log \epsilon^{-1} )$ iterations are sufficient to achieve $\mathrm{W}_2\left(\mu_k, \pi\right) \leq \alpha^{-1/2} \epsilon$, which matches previous results~\citep{sanz-serna_wasserstein_2021,dalalyan_sampling_2020,altschuler_shifted_2025}.
We make some concluding remarks in \cref{section:discussions}, while the proofs are deferred to \cref{section:proofs}.

\section{Preliminaries}\label{section:preliminaries}

\paragraph{Notation}
For some Euclidean space $\mathcal{X} \subseteq \mathbb{R}^d$, we denote its Borel-measurable subsets as $\mathcal{B}(\mathcal{X})$.
$\mathcal{P}_2\left(\mathcal{X}\right) = \{ \mu \mid \int_{\mathcal{X}} \norm{x}^2 \mu\left(\mathrm{d}x\right) < +\infty \}$ denotes the set of all distributions on $\mathcal{X}$ with a finite second moment. 
For vectors $x, y \in \mathbb{R}^d$, $\inner{x}{y} = x^{\top} y$ and $\norm{x} = \sqrt{\inner{x}{x}}$ denote the Euclidean inner product and norms, respectively.
Furthermore, for any random variable $X$, we denote the square root of its expectation as $\mathbb{E}^{1/2}X = \sqrt{\mathbb{E}X}$.
For a Markov kernel $Q : \mathcal{X} \times \mathcal{B}\left(\mathcal{X}\right) \to \mathbb{R}_{\geq 0}$ and a probability measure $\mu : \mathcal{B}\left(\mathcal{X}\right) \to \mathbb{R}_{\geq 0}$, their composition is denoted as $\mu Q\left(\mathrm{d}y\right) = \int_{\mathcal{X}} \mu\left(\mathrm{d}x\right) Q\left(x, \mathrm{d}y\right)$.
For a diagonalizable matrix $A \in \mathbb{R}^{d \times d}$ and any $p\in \{1, \ldots, d\}$, $\sigma_p\left(A\right)$ denotes its $p$th eigenvalue.

% Similarly, for any measurable function $\varphi$,  $Q f\left(x\right) = \int_{\mathcal{X}} Q\left(x, \mathrm{d}y\right) f\left(\mathrm{d}y\right)$.

\subsection{Stochastic Exponential Euler Discretization}
We begin with a high-level derivation of the stochastic exponential Euler discretization.
(Refer to \citet[Lem. 29]{durmus_uniform_2025} for a more rigorous derivation.)
Consider the fact that the solution to \cref{eq:kinetic_langevin_dynamics} at time \(t = T\) is given as the intractable expression
\begin{align}
    \begin{split}
    V_T
    &=
    \mathrm{e}^{- \gamma T} v_0 - \eta \int^T_0 \mathrm{e}^{- \gamma \left(T - t\right)} \nabla U\left(X_t\right) \mathrm{d}t + \sqrt{2 \gamma \eta} \int^T_0 \mathrm{e}^{-\gamma \left(T - t\right)} \, \mathrm{d}B_t
    \\
    X_T &= X_0 + \int^T_0 V_s \, \mathrm{d}t \; .
    \end{split}
    \label{eq:kinetic_langevin_solution}
\end{align}
If we hold the drift $\nabla U$ constant, \cref{eq:kinetic_langevin_solution} reduces to an Ornstein-Uhlenbeck process, which has a known solution.
The exponential integrator exploits this, for each $k \geq 0$, by integrating \cref{eq:kinetic_langevin_solution} over the interval $[hk, h (k + 1)]$ and replacing the state-dependent drift $\nabla U\left(X_t\right)$ with the state-independent constant drift $\nabla U(X_{hk})$.
Defining \(\delta \triangleq \mathrm{e}^{- \gamma h}\) for convenience, this yields the update rule for the discrete-time Markov chain \({(Z_k = (X_k, V_k))}_{k \geq  0}\),
\begin{align}
    \begin{split}
    X_{k + 1}
    &=
    X_{k} + \frac{1 - \delta}{\gamma} V_k - \eta \frac{\gamma h + \delta - 1}{\gamma^2} \nabla U\left(X_k\right) 
    + 
    \xi^X_{k+1}
    \\
    V_{k + 1}
    &=
    \delta V_{k} 
    - \eta \frac{1 - \delta}{\gamma} \nabla U\left(X_k\right) 
    + \xi^V_{k+1} 
    \; ,
    \end{split}
    \label{eq:kinetic_langevin_monte_carlo_algorithm}
\end{align}
%\zeta->\xi
where the sequence of noise variables ${(\xi^X_k, \xi^V_k)}_{k \geq 1}$ is given as
\begin{align*}
    \xi^X_{k+1}
    \triangleq
    \sqrt{2 \gamma \eta} \int^{h (k+1)}_{hk} \frac{1 - \mathrm{e}^{- \gamma (h (k+1) - s)}}{\gamma} \, \mathrm{d}B_{s} \, 
    \;\text{and}\;
    \xi^V_{k+1}
    \triangleq
    \sqrt{2 \gamma \eta} \int^{h (k+1)}_{h k} \mathrm{e}^{-\gamma \left(h (k+1) - s\right)} \, \mathrm{d}B_{s}
    \; .
\end{align*}
The noise sequence ${(\xi^X_k, \xi^V_k)}_{k \geq 1}$ can be simulated by drawing independent zero-mean $2d$-dimensional Gaussian random vectors with covariance 
$
\begin{bmatrix}
    \sigma_{XX}^2 & \sigma_{XV}^2  \\   
    \sigma_{XV}^2 & \sigma_{VV}^2
\end{bmatrix}
\otimes \mathrm{I}_d ,
$
where $\otimes$ denotes the Kronecker product, and
\begin{align*}
	\sigma_{XX}^2 = \frac{2 \eta}{\gamma} \;
	\left(h - 2 \frac{1 - \delta}{\gamma} + \frac{1 - \delta^2}{2 \gamma} \right) \; ,
    \quad
	\sigma_{XV}^2 = \frac{\eta}{\gamma} {\left( 1 - \delta \right)}^2 \; ,
    \quad
	\sigma_{VV}^2 = \eta \;
	\left( 1 - \delta^2 \right) \; .
\end{align*}
Throughout the paper, we denote the corresponding Markov kernel as $K : \mathbb{R}^{2 d} \times \mathcal{B}\left(\mathbb{R}^{2 d}\right) \to \mathbb{R}_{\geq 0}$ such that the Markov chain ${(Z_k)}_{k \geq 0}$ is simulated as $Z_{k+1} \sim K\left(Z_k, \cdot\right)$ for each $k \geq 0$.

\subsection{Assumptions on the Potential}
Our goal is to analyze the speed of approximately generating a sample from $\pi$, the stationary distribution of the continuous process (\cref{eq:kinetic_langevin_dynamics}), by simulating the Markov chain ${(Z_k)}_{k \geq 0}$ (\cref{eq:kinetic_langevin_monte_carlo_algorithm}).
This amounts to analyzing the speed in which ${(Z_k)}_{k \geq 0}$ converges to stationarity and the difference between $\pi_h$ and $\pi$.
These properties generally depend on the properties of the drift $\nabla U$, and in turn the potential $U$.
In this work, we consider the following assumption:

\begin{assumption}\label{asusmption:hessian_bounded}
    The potential \(U : \mathbb{R}^d \to \mathbb{R}\) is twice differentiable and there exist some $\alpha \in (0, +\infty)$ and $\beta \in [\alpha, +\infty)$ such that, for any \(x \in \mathbb{R}^{d}\), 
    \[
        \alpha \mathrm{I}_d \quad\preceq\quad \nabla^2 U\left(x\right) \quad\preceq\quad \beta \mathrm{I}_d \; .
    \]
\end{assumption}
The existence of $\alpha$, $\beta$ is equivalent to the density $x \mapsto \exp\left(-U(x)\right)/Z$ being $\alpha$-strongly log-concave and $\beta$-log-Lipschitz smooth, and the ratio $\kappa \triangleq \beta/\alpha$ is referred to as the condition number.
\cref{asusmption:hessian_bounded} has been widely used for analyzing the non-asymptotic complexity of approximate sampling schemes based on discretized kinetic Langevin diffusions~\citep{dalalyan_sampling_2020,cheng_underdamped_2018,sanz-serna_wasserstein_2021,leimkuhler_contraction_2024,monmarche_highdimensional_2021,foster_shifted_2021,shen_randomized_2019,liu_double_2023,altschuler_shifted_2025,chak_reflection_2025}.
In particular, it is known that \cref{asusmption:hessian_bounded} implies a contraction in Wasserstein distance for various discretizations of the kinetic/underdamped~\citep{leimkuhler_contraction_2024,monmarche_highdimensional_2021,sanz-serna_wasserstein_2021} and overdamped~\citep{durmus_highdimensional_2019} Langevin diffusions.
We are, however, interested in how precisely and quantitatively the discretized diffusion depends on the properties of the drift represented by $\alpha$ and $\beta$.

\vspace{-2ex}
\subsection{Parametrization}\label{section:parametrization}
\vspace{-1ex}
Notice that, in \cref{eq:kinetic_langevin_monte_carlo_algorithm}, the Markov kernel $K$ only interacts with the step size $h$ through the product $\zeta \triangleq \gamma h$.
% In fact, we will see in \cref{section:asymptotic_bias} that the value of $\zeta$ determines the regime of the damping.
Without loss of generality, it is possible to parametrize \cref{eq:kinetic_langevin_monte_carlo_algorithm} with $(\zeta, \gamma, \eta)$ instead of the usual $(h, \gamma, \eta)$.
Then $\delta$ becomes a monotonic transformation of $\zeta$ such that $\delta = \exp\left(-\zeta\right)$, while the update rule simplifies into
\begin{align}
    \begin{split}
	X_{k+1} &= X_k + \left( 1 - \delta \right) \left(\frac{1}{\gamma} V_k \right) - \left(\zeta + \delta - 1\right) \left(\frac{\eta }{\gamma^2} \nabla U\left(X_k\right) \right) + \xi^X_{k+1}
	\\
	V_{k+1} &= \delta V_k - \left(1 - \delta\right) \left( \frac{\eta}{\gamma} \nabla U\left(X_k\right) \right) + \xi^V_{k+1} \; .
    \end{split}
    \label{eq:klmc_zeta_update}
\end{align}
As a result, the analysis of the discretized algorithms is clearer and more natural in the $(\zeta, \gamma, \eta)$ parametrization.
Therefore, our Wasserstein contraction analysis will operate under this parametrization.
In the main text, however, we will present most results in the $(h, \gamma, \eta)$ parametrization to be consistent with the literature.

Also, in the $X_t$ update of \cref{eq:klmc_zeta_update}, notice that the gradient $\nabla U$ is scaled as $\eta/\gamma^2$.
This scaling naturally appears in the analysis through
\[
    R\left(\lambda\right) \triangleq \frac{\eta \lambda}{\gamma^2} \; ,
\]
where $\lambda \in [\alpha, \beta]$ represents any eigenvalue lying on the spectrum of $\nabla^2 U$.
Then, under \cref{asusmption:hessian_bounded}, $R$ is bounded as
\[
    \frac{\eta \alpha}{\gamma^2} \quad\leq\quad R\left(\lambda\right) \quad\leq\quad \frac{\eta \beta}{\gamma^2} \; .
\]
The contraction of the discretized dynamics is directly dependent on the behavior of the scaled eigenvalues $R\left(\lambda\right)$.
In fact, the scaling $\eta/\gamma^2$ partially hints at the fact that, for the discretized dynamics, keeping the ratio $\eta/\gamma^2$ constant results in similar behavior.
For instance, the choices of $\eta \asymp 1/\beta$ and $\gamma \asymp 1$~\citep{sanz-serna_wasserstein_2021} and $\eta \asymp 1$ and $\gamma \asymp \sqrt{\beta}$~\citep{dalalyan_sampling_2020,altschuler_shifted_2025} have been used to obtain mixing time complexities that are comparable.

\subsection{Weighted Norm and Wasserstein Distance}\label{section:norm}
For obtaining a tight contraction for the kinetic Langevin dynamics, it is necessary to consider the Wasserstein-2 distance induced by an unconventional norm~\citep{dalalyan_sampling_2020,ma_there_2021,monmarche_highdimensional_2021}.
Following previous works~ \citep{monmarche_highdimensional_2021,leimkuhler_contraction_2024}, we consider the following norm defined on the augmented state space $\mathbb{R}^{2 d}$, where, for any $z = (x,v) \in \mathbb{R}^{2 d}$, 
\[
    {\lVert z \rVert}_{a,b}^2 \;\;\triangleq\;\; {\lVert x \rVert}^2 + 2 b {\langle x, v \rangle} + a {\lVert v \rVert}^2,
\]
and $a, b \in \mathbb{R}_{>0}$ satisfy $b^2 < a$. 
The last condition ensures that $ {\lVert z \rVert}_{a,b}^2$ is a valid norm, which can be verified by using Young's inequality (for any $\epsilon > 0$, we have $2 b\langle x,v\rangle \leq \epsilon {\lVert x \rVert}^2 +  b^{2}\epsilon^{-1} {\lVert v \rVert}^2$).
Furthermore, if we enforce the stronger condition $4 b^2 \leq a$, we retrieve an explicit equivalence with the conventional Euclidean norm as
\begin{align}
    \frac{1}{2} {\lVert z \rVert}_{a,0}^2 \quad\leq\quad {\lVert z \rVert}_{a,b}^2 \quad\leq\quad \frac{3}{2} {\lVert z \rVert}_{a,0}^2 \; .
    \label{eq:norm_equivalence}  
\end{align}
For the values of the norm coefficients $a$ and $b$, we will use the specific values of $a = 4/\gamma^2$ and $b = 1/\gamma$, which satisfy the condition for \cref{eq:norm_equivalence}.

We denote Wasserstein-2 distance induced by the norm $\norm{\cdot}_{a,b}$ as
\[
    \mathrm{W}_{a,b}\left(\mu, \nu\right)
    \triangleq
    \inf_{\rho \in \Gamma(\mu, \nu)} \sqrt{ \int_{\mathbb{R}^{2d} \times \mathbb{R}^{2d}} \norm{z - z'}_{a,b}^2 \mathrm{d}\rho\left(z, z'\right) } \; ,
\]
where \(\Gamma(\mu, \nu)\) is the set of couplings between \(\mu \in \mathcal{P}(\mathbb{R}^{2d})\) and \(\nu \in \mathcal{P}(\mathbb{R}^{2d})\).
The conventional Wasserstein-2 distance can be correspondingly defined as $W_2 \triangleq W_{1,0}$.

% In contrast, \citet[Thm. 5]{cheng_underdamped_2018} and \citet[Thm. 6.1]{sanz-serna_wasserstein_2021} used $a = 1$, $b = 1/2$, \citet[Prop. 17]{monmarche_highdimensional_2021} and \citet[Thm 6.1]{leimkuhler_contraction_2024} used $a = 1/\beta$ and $b = 1/\gamma$, \citet[Thm 5.3.5]{chewi_logconcave_2024} uses $a = 2/\gamma^2$ and $b = 1/\gamma$, while \citet{dalalyan_sampling_2020} chose $a$ and $b$ adaptively depending on the condition number.

\section{Main Results}

\begin{figure*}
    \centering
    \floatsetup{heightadjust=object, valign=t}
    \ffigbox{
        \begin{subfloatrow}
            \ffigbox[\FBwidth]{
                \begin{tikzpicture}
    \begin{axis}[
        xmin=0,
        xmax=0.5,
        ymin=0,
        ymax=0.5,
        ylabel={},
        major tick length=2pt,
        minor tick length=0pt,
        xtick align=outside,
        ytick align=outside,
        ytick pos=left, 
        xtick pos=bottom, 
        axis line style = thick,
        every tick/.style={black,thick},
        xtick style={draw=none},
        ytick style={draw=none},
        xtick = {0.1, 0.2726},
        ytick = {0.17797},
        xticklabels = {}, %{$\frac{\eta \mu}{\gamma^2}$, $\frac{\eta L}{\gamma^2}$},
        yticklabels = {$c$},
        width = 36ex,
        height= 33ex,
    ]
    
    \addplot[restrict x to domain=0.095:0.2726, samples=100, fill=gray!30, draw=none, mark=none] table[x index=0, y index=2, col sep=comma] {figures/cstar.txt} \closedcycle;

    %restrict x to domain=0:0.445
    \addplot[cherry3, mark=none, very thick, name path=cstar] table[x index=0, y index=2, col sep=comma] {figures/cstar.txt} node[pos=0.2, above] {\small$p_1 - \sqrt{p_2 \, p_3}$};
    
    \path [name path=c] (axis cs: 0, 0.17797) -- (axis cs: 1, 0.17797);
    \path [name intersections={of=cstar and c}];
    \draw[cherry3, thick, densely dotted] (0, 0.17797) -- (intersection-1); 

    \end{axis}
\end{tikzpicture}
            }{\vspace{-2ex}\caption{\cref{thm:general_contraction}}}
            \ffigbox[\FBwidth]{
                \begin{tikzpicture}
    \begin{axis}[
        xmin=0,
        ymin=0,
        xlabel={Eigenvalue of $\nabla^2 U$},
        ylabel={\footnotesize$p_1 - \sqrt{p_2 \, p_3}$},
        major tick length=2pt,
        minor tick length=0pt,
        xtick align=outside,
        ytick align=outside,
        ytick pos=left, 
        xtick pos=bottom, 
        axis line style = thick,
        every tick/.style={black,thick},
        xtick style={draw=none},
        ytick style={draw=none},
        xtick = {0.1},
        ytick = {0.1},
        xticklabels = {},
        yticklabels = {},
        width = 36ex,
        height= 33ex,
        legend cell align={left},
        legend style={
            at={(0.97,0.97)},
            anchor=north east,
            draw=none,
            fill opacity=0.8,
            text opacity=1.0,
        },
    ]
    
    %restrict x to domain=0:0.445
    \addplot[cherry6, mark=none, very thick, name path=cstar] table[x index=0, y index=1, col sep=comma] {figures/cstar.txt};
    \addlegendentry{$h \gamma = 0.5$}
    
    \addplot[cherry4, mark=none, very thick, name path=cstar] table[x index=0, y index=2, col sep=comma] {figures/cstar.txt};
    \addlegendentry{$h \gamma = 1$}
    
    \addplot[cherry2, mark=none, very thick, name path=cstar] table[x index=0, y index=3, col sep=comma] {figures/cstar.txt};
    \addlegendentry{$h \gamma = 2$}

    \end{axis}
\end{tikzpicture}
            }{\vspace{-2ex}\caption{Effect of $h \gamma$ on \cref{thm:general_contraction}\label{fig:cstar_zeta}}}
            \ffigbox[\FBwidth]{
                \begin{tikzpicture}
    \begin{axis}[
        xmin=0,
        xmax=0.2,
        ymin=0,
        ymax=0.5,
        ylabel={},
        major tick length=2pt,
        minor tick length=0pt,
        xtick align=outside,
        ytick align=outside,
        ytick pos=left, 
        xtick pos=bottom, 
        axis line style = thick,
        every tick/.style={black,thick},
        xtick style={draw=none},
        ytick style={draw=none},
        xtick = {0.05, 0.1437},
        ytick = {0.05},
        xticklabels = {}, %{$\frac{\eta \mu}{\gamma^2}$, $\frac{\eta L}{\gamma^2}$},
        yticklabels = {$c$},
        width = 36ex,
        height= 33ex,
    ]
    
    \addplot+[mark=none, domain=0.05:0.1437, samples=100, draw=none, fill=gray!30] {x} \closedcycle;

    %restrict x to domain=0:0.445
    \addplot[cherry1, mark=none, very thick, name path=cstar] table[x index=0, y index=2, col sep=comma] {figures/cstar.txt} node[pos=0.1, above, yshift=2ex] {\small$p_1 - \sqrt{p_2 p_3 }$};
    
    % \path [name path=cstarthres] (axis cs: 0.44, 0) -- (axis cs: 0.44, 1);
    % \path [name intersections={of=cstar and cstarthres}];
    % \draw[cherry1, thick, dashed] (0.44, 0) -- (intersection-1); 
    
    \addplot[cherry3, mark=none, very thick, domain=0:0.1437, name path=clin] {x} node[pos=1.0, anchor=west] {$h \gamma R$};
    
    % \path [name path=clinthres] (axis cs: 0.2726, 0) -- (axis cs: 0.2726, 1);
    % \path [name intersections={of=clin and clinthres}];
    % \draw[cherry3, thick, dashed] (0.2726, 0) -- (intersection-1); 
    
    \path [name path=c] (axis cs: 0, 0.05) -- (axis cs: 1, 0.05);
    \path [name intersections={of=clin and c}];
    \draw[cherry3, thick, densely dotted] (0, 0.05) -- (intersection-1); 

    \end{axis}
\end{tikzpicture}
            }{\vspace{-2ex}\caption{\cref{thm:special_case_convergence_bound}}\label{fig:clin}}
        \end{subfloatrow}
    }{\caption{
        \textbf{Illustration of \cref{thm:general_contraction,thm:special_case_convergence_bound}}.
        The grey region represents the spectrum of $\nabla^2 U$ under \cref{asusmption:hessian_bounded} for a certain choice of parameters.
        Intuitively, the grey region becomes wider as the problem becomes less well-conditioned (larger $\kappa = \beta/\alpha$).
        (a) Relationship between the function $p_1 - \sqrt{p_2 \, p_3}$, spectrum of $\nabla^2 U$, and the contraction coefficient $c$.
        (b) Increasing $\zeta = h \gamma$ raises the peak value of $p_1 - \sqrt{p_2 \, p_3}$ but reduces the range of $R$ where $p_1 - \sqrt{p_2 \, p_3}$ is positive.
        This results in a trade-off between the condition number $\kappa = \beta / \alpha$ and the resulting contraction coefficient.
        (c) Visualization of the linear under-approximation of $p_1 - \sqrt{p_2 \, p_3}$.
        Notice that the resulting contraction coefficient becomes worse.
    }}\label{fig:contraction}
\end{figure*}
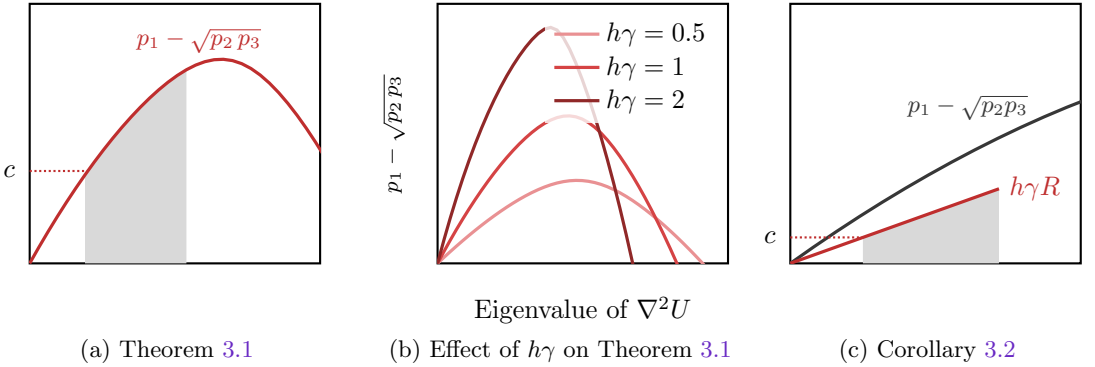

\subsection{Convergence to Stationarity}\label{section:convergence_stationarity}

Firstly, we present a general Wasserstein contraction result for the Markov kernel $K$ associated with the kinetic Langevin dynamics discretized via the exponential integrator.
This will immediately imply that $K$ admits a stationary distribution $\pi_h$ and that it converges to $\pi_h$ at a dimension-independent geometric rate.
We follow the strategy~\citep{leimkuhler_contraction_2024,sanz-serna_wasserstein_2021} of reducing the problem to solving a special case of the discrete-time Lyapunov equation~\citep[\S 6.E]{antsaklis_linear_2006}.
Similarly to \citet{leimkuhler_contraction_2024}, we directly solve the Lyapunov equation by analyzing the eigenvalues of a collection of matrices.
However, our analysis differs in that we first obtain the exact expression for the contraction rate.
Under \cref{asusmption:hessian_bounded} and appropriate conditions on the parameters $h, \gamma, \eta$, $K$ admits a contraction in Wasserstein distance induced by the weighted norm $\norm{\cdot}_{a,b}$ (defined in \cref{section:norm}) using $a = 4/\gamma^2$ and $b = 1/\gamma$.

\begin{theoremEnd}[%
    text link={\noindent\textit{Proof.} The proof is deferred to \cref{section:proof_general_contraction}. \qed},
    category=generalcontraction, 
    text proof = {Proof of \string\pratendRef{thm:prAtEnd\pratendcountercurrent}}
]{theorem}\label{thm:general_contraction}
    Suppose \cref{asusmption:hessian_bounded} holds and the parameters $h, \gamma, \eta$, where $\delta = \exp(-h\gamma)$, satisfy
    \begin{equation}
        \eta
        \left(
            \frac{2}{3}
            \frac{h}{ \gamma {\left(1 - \delta^2\right)}}
            +
            \frac{3}{2}
            \frac{1}{\gamma^2} 
        \right)
        \quad\leq\quad
        \frac{1}{\beta}
        \; .
        \label{eq:condition_general_contraction}
    \end{equation}
    Then, for any \(\mu, \nu \in \mathcal{P}_2\left(\mathbb{R}^{2 d}\right)\) and all $n \geq 1$, we have
    \[
        {\mathrm{W}_{a,b}\left(\mu K^{n}, \nu K^{n}\right)}^2 \quad\leq\quad  {\left(1 - c\left(h, \gamma, \eta\right) \right)}^n \; 
        {\mathrm{W}_{a,b}\left(\mu, \nu\right)}^2
    \]
    with the contraction coefficient 
    \[
        c\left(h, \gamma, \eta\right) = \inf_{\lambda \in \left[\alpha, \beta\right]} \quad p_1\left( R\left(\lambda\right) \right) - \sqrt{ p_2\left( R\left(\lambda\right) \right) p_3\left( R\left(\lambda\right) \right)} \; ,
    \]
    which is strictly positive, where
    \begin{align*}
        p_1\left(r\right) \triangleq -a_1 r^2 + b_1 r + e_1 \; ,  \qquad
        p_2\left(r\right) \triangleq a_1 r^2 + b_2 r + e_2 \; ,  \qquad
        p_3\left(r\right) \triangleq a_1 r^2 + b_3 r + e_3 \; ,  \qquad
    \end{align*}
with the coefficients 
\begin{alignat*}{5}
	a_1 &\triangleq \frac{2}{3} {(h \gamma)}^2 + 2 {\left(1 - \delta\right)}^2 \; ,
	\qquad
	&&b_1 &&\triangleq h \gamma - \left(\delta - \delta^2\right) \; ,
	\qquad
	&&e_1 &&\triangleq \frac{1}{2} \left(1 - \delta^2\right) \; ,
	\qquad
	\\
	&
	\qquad
	&&b_2 &&\triangleq - h \gamma \left(1 + \delta\right) + \left(1 - \delta^2\right) \; ,
	\qquad
	&&e_2 &&\triangleq \frac{1}{2} {\left(1 + \delta\right)}^2 \; ,
	\\
	&
	\qquad
	&&b_3 &&\triangleq - h \gamma \left(1 - \delta\right) - {\left(1 - \delta\right)}^2 \; ,
	\qquad
	&&e_3 &&\triangleq \frac{1}{2} {\left(1 - \delta\right)}^2 \; .
\end{alignat*}
\end{theoremEnd}
\begin{proofEnd}
Recall $r_{\mathrm{max}}$ in \cref{eq:rmax}. 
\cref{thm:ACmB2_positive_definite} implies
\[
    r \in (0, r_{\mathrm{max}}) \qquad\Rightarrow\qquad c^{-}\left(r, \zeta\right) > 0 \; .
\]
Furthemore, as long as $c^-\left(r, \zeta\right) > 0$ can be ensured, 
\begin{align*}
    c \in (0, c^-\left(r, \zeta\right)]
    \qquad&\Rightarrow\qquad 
    \chi_{AC-B^2}\left(r, \zeta, \gamma, c\right) &&\geq 0
    \quad\text{and}
    &&\text{(\cref{thm:ACmB2_nonnegative})}
    \\
    c \in (0, c^-\left(r, \zeta\right)]
    \qquad&\Rightarrow\qquad  
    \chi_{A}\left(r, \zeta, \gamma, c\right) &&> 0 \; .
    &&\text{(\cref{thm:A_positive_definite})}
\end{align*}
%Now recall from the discussion in \cref{section:parametrization} that $R\left(\lambda\right)$ represents the scaled eigenvalues of $\nabla U^2$. 
Recall the equivalence in \cref{eq:chi_spectrum_equivalence} that $\chi_A$ and $\chi_{AC-B^2}$ can be related to the spectrum of $A_k$, $A_k C_k - B_k^2$ through the choice of $r = R\left(\sigma\left(H_k\right)\right)$.
Furthermore, under \cref{asusmption:hessian_bounded}, 
\[
    0 \quad<\quad \frac{\alpha\eta}{\gamma^2} \quad\leq\quad R\left(\sigma_p\left(H_k\right)\right) \quad\leq\quad \frac{\beta \eta}{\gamma^2} \; .
\]
Therefore, the choice of contraction coefficient $c$ in the proof statement satisfies
\[
    c = 
    \inf_{\lambda \in [\alpha, \beta]} c^-\left( R\left(\lambda\right), \zeta \right) 
    \qquad\Rightarrow\qquad
    \forall k \geq 0, \quad  c \in \left( 0,  c^-\left(R\left( \sigma_p\left(H_k\right) \right), \zeta\right) \right)
\]
On the other hand, the boundedness of $r = R\left(\sigma\left(H_k\right)\right)$ can be ensured through
\[
    \frac{\beta \eta}{\gamma^2} < r_{\mathrm{max}}
    \quad\Rightarrow\quad
    \forall k \geq 0, \;\;
    0 < R\left(\sigma_p\left(H_k\right) \right) < r_{\mathrm{max}}
    \quad\Rightarrow\quad
    \forall k \geq 0, \;\;
    r \in (0, r_{\mathrm{max}})
\]
Since the expression for $r_{\mathrm{max}}$ is rather complex, the condition in \cref{eq:condition_general_contraction} serves as a simpler sufficient condition.
This follows from the implications
\begin{align}
    & &
    \frac{\beta \eta}{\gamma^2} 
    \quad&<\quad
    r_{\mathrm{max}}
    \nonumber
    \\
    &\Leftrightarrow&
    \frac{\beta \eta}{\gamma^2} 
    \quad&<\quad
    \frac{
    2 \zeta \left(1 - \delta^2\right)
    }{
        \left(\frac{4}{3} - \delta^2\right) \zeta^2  + 2 \delta \left(1 - \delta\right) \zeta + 3 \left(1 - \delta\right)^2
    } 
    \nonumber
    \\
    &\Leftrightarrow&
    \frac{1}{\beta}
    \quad&>\quad
    \frac{\eta}{\gamma^2} 
    \frac{
        \left(\frac{4}{3} - \delta^2\right) \zeta^2  + 2 \delta \left(1 - \delta\right) \zeta + 3 \left(1 - \delta\right)^2
    }{
        2 \zeta \left(1 - \delta^2\right)
    }
    \nonumber
    \\
    &\Leftrightarrow&
    \frac{1}{\beta}
    \quad&>\quad
    \frac{\eta}{\gamma^2} 
    \frac{
        \left(\frac{4}{3} - \delta^2\right) \zeta  + 2 \delta \left(1 - \delta\right)  + 3 \left(1 - \delta\right)^2 / \zeta
    }{
        2 \left(1 - \delta^2\right)
    }
    \nonumber
    \\
    &\Leftarrow&
    \frac{1}{\beta}
    \quad&\geq\quad
    \frac{\eta}{\gamma^2} 
    \frac{
        \left(\frac{4}{3} - \delta^2\right) \zeta  + 3 \delta \left(1 - \delta\right)  + 3 \left(1 - \delta\right) 
    }{
        2 \left(1 - \delta^2\right)
    }
    &&\text{($1 - \delta \leq \zeta$, $2 \delta < 3 \delta $)}
    \nonumber
    \\
    &\Leftrightarrow&
    \frac{1}{\beta}
    \quad&\geq\quad
    \frac{\eta}{\gamma^2} 
    \frac{
        \left(\frac{4}{3} - \delta^2\right) \zeta  + 3 \left(1 - \delta^2\right)
    }{
        2 \left(1 - \delta^2\right)
    }
    \nonumber
    \\
    &\Leftarrow&
    \frac{1}{\beta}
    \quad&\geq\quad
    \frac{\eta}{\gamma^2} 
    \frac{
        \frac{4}{3} \zeta + 3 {\left(1 - \delta^2\right)}
    }{
        2 {\left(1 - \delta^2\right)}
    }
    && \text{($-\delta^2 < 0$)}
    \nonumber
    \\
    &\Leftrightarrow&
    &\text{\cref{eq:condition_general_contraction}}
    \nonumber
    \; .
\end{align}
Therefore, under the choices of $r = R\left(\sigma_p\left(H_k\right)\right)$, $c = \inf_{\lambda \in [\alpha, \beta]} c^-\left(R\left(\lambda\right), \zeta\right)$, 
\[
    \text{\cref{eq:condition_general_contraction}}
    \qquad\Rightarrow\qquad
    \forall k \geq 0, \quad
    A_k \succ 0 
    \quad\text{and}\quad
    A_k C_k - B_k^2 \succeq 0  \; .
\]
Finally, by invoking \cref{thm:lyapunov_equation_block_analysis}, we can conclude that, as long as \cref{eq:condition_general_contraction} is ensured, \cref{eq:one_step_contraction} holds with $c = \inf_{\lambda \in [\alpha, \beta]} c^-\left(r, \zeta\right) > 0$ for all $k \geq 0$.

We now know that, under \cref{eq:condition_general_contraction}, the one-step contraction in \cref{eq:one_step_contraction} holds with the coefficient $c = \inf_{\lambda \in [\alpha, \beta]} c^-\left(r, \zeta\right)$.
This implies that, for any $n \geq 1$, unrolling the recursion over $k = 1, \ldots, n$ yields that
\[
    \norm{Z_{n} - Z_{n}'}_{a,b}^2 \leq {\left(1 - c\right)}^n \norm{Z_0 - Z_0'}_{a,b}^2
\]
holds almost surely.
Now, for any $\mu, \nu \in \mathcal{P}_2\left(\mathbb{R}^{2d}\right)$, we know that there exists an optimal coupling $\rho^*$ between the two~\citep[Cor. 5.22]{villani_optimal_2009}.
By initializing $(Z_0, Z_0') \sim \rho^*$ and taking expectation over the noise sequence ${(\xi^X_{k}, \xi^V_k)}_{k \geq 1}$, for any $n \geq 1$, we obtain
\[
    {\mathrm{W}_{a,b}\left(\mu K^n, \nu K^n\right)}^2
    \;\leq\;
    \mathbb{E}\norm{Z_{n} - Z_{n}'}_{a,b}^2 
    \;\leq\;
    {\left(1 - c\right)}^n \mathbb{E} \norm{Z_0 - Z_0'}_{a,b}^2 
    \;=\;
    {\left(1 - c\right)}^n \mathrm{W}_{a,b}\left(\mu, \nu\right)
    \; .
\]
\end{proofEnd}

This immediately implies the existence of a unique stationary distribution.

\begin{corollary}
    Suppose the conditions of \cref{thm:general_contraction} hold.
    Then $K$ also admits a unique stationary distribution $\pi_h \in \mathcal{P}_2\left(\mathbb{R}^{2 d}\right)$.
\end{corollary}
\begin{proof}
    $(\mathcal{P}_2\left(\mathbb{R}^{2 d}\right), \mathrm{W}_{2})$ forms a complete metric space metrized by $\mathrm{W}_{2}$~\citep[Thm. 6.18]{villani_optimal_2009}.
    The result then follows from \cref{thm:general_contraction}, the equivalence between $\norm{\cdot}_{a,b}$ and the Euclidean norm under $a = 4/\gamma^2$, $b = 1/\gamma$, and the Banach fixed point theorem.
\end{proof}

For the special case of $\nu = \pi_h$, \cref{thm:general_contraction} immediately implies a geometric rate of convergence of ${(Z_k)}_{k \geq 0}$ to its stationary distribution $\pi_h$.
The main contribution of \cref{thm:general_contraction}, however, is the weaker restriction on the parameters stated in \cref{eq:condition_general_contraction}, where we notice an interplay between the parameters $h, \gamma, \eta$.
In the overdamped regime, 
\[
    \eta \left( \frac{2}{3} \frac{h}{ \gamma \left( 1 - \delta^2 \right) } + \frac{3}{2} \frac{1}{\gamma^2} \right) 
    \qquad\stackrel{\gamma \to \infty}{\simeq}\qquad
    \frac{2}{3} \frac{\eta h}{\gamma} 
    \; .
\]
Therefore, the more we increase $\gamma$, the more we are free to increase either $\eta$ or $h$.
This contrasts with the result by \citet[Thm. 6.1]{leimkuhler_contraction_2024}, which strictly requires $h \leq 1/(2 \gamma)$.
In the underdamped regime, the inequality
$
    \frac{\zeta}{1 - \mathrm{e}^{-\zeta}} \leq 1 + \frac{\zeta}{2} + \frac{\zeta^2}{6}
$
yields
\[
    \eta \left( \frac{2}{3} \frac{h}{ \gamma \left( 1 - \delta^2 \right) } + \frac{3}{2} \frac{1}{\gamma^2} \right) 
    \quad\leq\quad
    \eta \left( 
    \frac{11}{6} \frac{1}{\gamma^2} + \frac{1}{3} \frac{h}{\gamma} + \frac{2}{9} h^2
    \right) \; .
\]
That is, in the underdamped regime, $\gamma \geq \Omega(\sqrt{\beta})$ and $h \leq \mathrm{O}(1/\sqrt{\beta})$, which is analogous to the constraint $h \leq \mathrm{O}(1/\gamma)$.
Thus, our restriction on the parameters does agree with that of \citet[Thm. 6.1]{leimkuhler_contraction_2024} in the underdamped regime.
Furthermore, for the fixed choices of $\gamma = 2$ and $\eta = 1/\beta$, as taken by~\citet{sanz-serna_wasserstein_2021}, \cref{eq:condition_general_contraction} becomes 
\[
     \frac{h}{1 - \exp\left(-4 h\right)} \quad\leq\quad \frac{15}{8} \; .
\]
Numerically solving this inequality yields an approximate condition of $h \leq 1.87$, which is slightly weaker than the $h \leq 1$ condition by~\citet[Thm. 6.1]{sanz-serna_wasserstein_2021}.
% By imposing slightly stronger assumptions on the parameters, we obtain a more interpretable contraction coefficient.
% The following result follows by performing a linear under-approximation of the function $p_1 - \sqrt{p_2 p_3}$.

Meanwhile, the contraction coefficient $c(h, \gamma, \eta)$ is determined by the function $p_1 - \sqrt{p_2 p_3}$.
While $p_1 - \sqrt{p_2 p_3}$ is difficult to interpret, it is an exact expression for the contraction coefficient.
The behavior of $c(h, \gamma, \eta)$ depending on the parameters is illustrated in \cref{fig:contraction}.
To obtain a more interpretable expression for the contraction coefficient, we can perform a linear under-approximation of $p_1 - \sqrt{p_2 p_3}$ (illustrated in \cref{fig:clin}) with slightly stronger restrictions on the parameters:

\begin{assumption}\label{assumption:condition_linear_approximation_contraction}
    The parameters $h, \gamma, \eta > 0$ satisfy the inequality 
    \[
        \eta
        \left(
        2
        \frac{h}{ \gamma {\left(1 - \delta\right)}}
        +   
        \frac{6}{\gamma^2} 
        \right) 
        \quad<\quad
        \frac{1}{\beta} \; .
    \]
\end{assumption}
We will rely on this condition throughout the remainder of the article.
It is apparent that \cref{assumption:condition_linear_approximation_contraction} is very similar to \cref{eq:condition_general_contraction}.
As such, in the overdamped regime, \cref{assumption:condition_linear_approximation_contraction} also reduces to the condition $h \eta / \gamma \leq \mathrm{O}\left( 1 / \beta \right) $.
For the underdamped regime, by relying on the inequality 
$
    \frac{\zeta}{1 - \mathrm{e}^{-\zeta}} \leq 1 + \frac{\zeta}{2} + \frac{\zeta^2}{6} 
$
again, it can be simplified through
\begin{align}
    \frac{\eta}{\gamma^2} \left( 2 \frac{h \gamma}{{\left(1 - \delta\right)}} + 6 \right) 
    \leq \frac{1}{\beta} 
    \quad&\Leftarrow\quad
    \frac{\eta}{\gamma^2} \left( 2 \left(1 + \frac{h \gamma}{2} + \frac{{(h \gamma)}^2}{6}\right) + 6 \right) 
    \leq \frac{1}{\beta} 
    \nonumber
    \\
    % \quad&\Leftrightarrow\quad
    % 8 + \zeta + \frac{\zeta^2}{3}
    % \quad\leq\quad 
    % \frac{\gamma^2}{\beta \eta} 
    % \nonumber
    % \\
    % \quad&\Leftarrow\quad
    % 8 + \zeta^2 + 1 + \frac{\zeta^2}{3}
    % \quad\leq\quad \frac{3 \gamma^2}{\beta \eta} 
    % \nonumber
    % \\
    % \quad&\Leftarrow\quad
    % \zeta
    % \leq \sqrt{ \frac{3}{4} \frac{\gamma^2}{\beta \eta} - \frac{27}{4} }
    % \quad\text{and}\quad 
    % \frac{\gamma^2}{\beta \eta} > 9
    % \nonumber
    % \\
    \quad&\Leftarrow\quad
    \gamma > \sqrt{9 \beta \eta}
    \quad\text{and}\quad 
    h \leq \sqrt{ \frac{3}{4} \frac{1}{\beta \eta} - \frac{27}{4} \frac{1}{\gamma^2} } 
    \; .
    \nonumber
\end{align}
For $\eta = 1$, this implies that $\gamma$ must satisfy at least $\gamma \geq \sqrt{9 \beta}$.
Up to constants, this corresponds to the result by \citet[Thm. 6.1]{leimkuhler_contraction_2024}.
On the other hand, choosing $\eta = 1/(2 \beta)$ and $\gamma = \sqrt{27/2}$ conveniently implies a condition of $h \leq 1$.
This covers the choice of $\gamma \asymp 1$ and $\eta \asymp 1/\beta$ by \citet[Thm 4.9]{sanz-serna_wasserstein_2021}.

Given \cref{assumption:condition_linear_approximation_contraction}, we obtain a simpler expression for the contraction coefficient.

\begin{theoremEnd}[%
    category=specialcaseconvergencebound,
    text proof = {Proof of \string\pratendRef{thm:prAtEnd\pratendcountercurrent}},
    text link={\noindent\textit{Proof.} The proof is deferred to \cref{section:proof_special_case_convergence_bound}. \qed}
]{corollary}\label{thm:special_case_convergence_bound}
Suppose \cref{asusmption:hessian_bounded,assumption:condition_linear_approximation_contraction} hold.
Then the contraction coefficient in \cref{thm:general_contraction} satisfies
\[
    c\left(h, \gamma, \eta\right) \quad\geq\quad \tilde{c}\left(h, \gamma, \eta\right) \quad\triangleq\quad \frac{h \eta \alpha}{\gamma}\; .
\]
\end{theoremEnd}
\begin{proofEnd}
Recall $r_{\mathrm{lin}}$ in \cref{eq:rlin}.
For all $r \in (0, r_{\mathrm{lin}}]$, we have the following equivalences:
\begin{align*}
    & &
    c^-\left(r, \zeta\right) &\geq \zeta r
    \\
    &\Leftrightarrow&
    p_1\left(r\right) - \sqrt{p_2\left(r\right) p_3\left(r\right)} &\geq \zeta r
    \\
    &\Leftrightarrow&
    {\left( p_1\left(r\right) - \zeta r \right)}^2 &\geq p_2\left(r\right) p_3\left(r\right)
    &&\text{(\cref{thm:p1_larger_than_zeta_r})}
    \\
    &\Leftrightarrow&
    p_6\left(r\right) r &\geq 0
    &&\text{(\cref{thm:p6_non_negative})} \; .
\end{align*}
This implies that, for the choice $r = R\left(\lambda\right)$, where $\lambda \in [\alpha, \beta]$, if $R\left(\lambda\right) \leq r_{\mathrm{lin}}$ can be ensured, the contraction coefficient in \cref{thm:general_contraction} is lower bounded as
\[
    \inf_{\lambda \in [\alpha, \beta]} c^-\left(R\left(\lambda\right), \zeta\right)
    \quad\geq\quad
    \inf_{\lambda \in [\alpha, \beta]} \zeta R\left(\lambda\right)
    \quad=\quad
    \zeta \frac{\eta \alpha}{\gamma^2} \; .
\]
Since $\lambda \in [\alpha, \beta]$, the condition $R\left(\lambda\right) \leq r_{\mathrm{lin}}$ is sufficiently ensured by \cref{assumption:condition_linear_approximation_contraction} due to the implications
\begin{align*}
    & &
    R\left(\lambda\right)
    \quad&\leq\quad
    r_{\mathrm{lin}}
    \quad=\quad
    \frac{
        \zeta {\left(1 - \delta\right)}
    }{
        2 \zeta^2  + 6 {\left(1 - \delta\right)}^2
    } 
    \\
    &\Leftarrow&
    \frac{\beta \eta}{\gamma^2} 
    \quad&\leq\quad
    \frac{
        \zeta {\left(1 - \delta\right)}
    }{
        2 \zeta^2  + 6 {\left(1 - \delta\right)}^2
    } 
    \\
    &\Leftrightarrow&
    \frac{1}{\beta}
    \quad&\geq\quad
    \frac{\eta}{\gamma^2} 
    \frac{
        2 \zeta^2  + 6 {\left(1 - \delta\right)}^2
    }{
        \zeta {\left(1 - \delta\right)}
    }
    \\
    &\Leftrightarrow&
    \frac{1}{\beta}
    \quad&\geq\quad
    \frac{\eta}{\gamma^2} 
    \left(
    2 \frac{\zeta}{{\left(1 - \delta\right)}}
    +   
    6
    \frac{1 - \delta}{\zeta}
    \right)
    \\
    &\Leftarrow&
    \frac{1}{\beta}
    \quad&\geq\quad
    \frac{\eta}{\gamma^2} 
    \left(
    2 \frac{\zeta}{{\left(1 - \delta\right)}}
    +   
    6
    \right)
    &&\text{($1 - \delta < \zeta $)}
    \; .
\end{align*}
\end{proofEnd}

This implies a $\mathrm{O}\left(h \eta \alpha/\gamma\right)$ contraction rate, which matches previous results~\citep{sanz-serna_wasserstein_2021,leimkuhler_contraction_2024}.
%
% To operate KLMC in the underdamped regime, \cref{eq:condition_linear_approximation_contraction} can be satisfied with $\eta = 1$, $\gamma = \sqrt{7}L$ and $\zeta \leq 1$.
%
However, our main contribution is the generality of the condition \cref{assumption:condition_linear_approximation_contraction}.
Specifically, it allows for the contraction to hold even in the overdamped limit:

\begin{corollary}[Overdamped Limit]\label{thm:convergence_overdamped}
    Suppose \cref{asusmption:hessian_bounded} holds and, for some $h_{\mathrm{LMC}} > 0$ satisfying $h_{\mathrm{LMC}} \leq 1/(2\beta)$, the parameters are set as $h = h_{\mathrm{LMC}} \, \gamma$, $\eta = 1$, and $\gamma \to \infty$.
    Then \cref{assumption:condition_linear_approximation_contraction} is satisfied while the overdamped limit of the contraction coefficient $c\left(h,\gamma,\eta\right)$ in  \cref{thm:special_case_convergence_bound} follows as
    \[
        \lim_{\gamma \to \infty} c\left(h, \gamma, \eta\right)
        \quad=\quad
        \lim_{\gamma \to \infty} c\left(h_{\mathrm{LMC}} \gamma, \gamma, 1\right)
        \quad=\quad
        h_{\mathrm{LMC}} \alpha \; .
    \]
\end{corollary}

This is equivalent to the Wasserstein contraction rate and step size limit of the Euler--Maruyama discretization of the overdamped Langevin dynamics~\citep[Prop 2.]{durmus_highdimensional_2019}.
Therefore, under appropriate scaling of the step size ($h \propto \gamma$), the exponential integrator is able to non-degenerately simulate the kinetic Langevin dynamics in the overdamped regime.
% This puts the exponential integrator in the class of ``$\gamma$-limit convergent'' integrators~\citep{leimkuhler_contraction_2024}, which includes the OBABO and BAOAB splitting schemes~\citep{leimkuhler_rational_2013}.

\subsection{Asymptotic Bias}\label{section:asymptotic_bias}

We now turn to analyzing the asymptotic bias $\mathrm{W}_{a,b}\left(\pi_h, \pi\right)$ of the stationary distribution of $K$, $\pi_h$.
Denote the Markov semigroup associated with \cref{eq:kinetic_langevin_dynamics} by ${(P_{t})}_{t \geq 0}$, where we recall that $\pi P_t = \pi$ and
${(Z_t^* \sim \pi P_t)}_{t \geq 0}$ is the kinetic Langevin dynamics initialized from its stationary distribution $\pi$. 
Given a discretized process ${(Z_k)}_{k \geq 0}$, which need not be stationary, our goal is to bound the distance between ${(Z_t^*)}_{t \geq 0}$ and ${(Z_k)}_{k \geq 0}$ as $t = h k \to \infty$.
The proof strategy follows~\citet{sanz-serna_wasserstein_2021}, where we construct an auxiliary process ${(Z^{\prime})}_{t \geq 0}$ that corresponds to the exponential integrator discretization of the $\pi$-stationary process ${(Z^*_t)}_{t \geq 0}$.
%and the corresponding stationary Markov process simulating \cref{eq:kinetic_langevin_dynamics} denoted as ${(\pi P_t)}_{t \geq 0}$.
%Also, denote the exponential integrator discretization applied to ${(P_t)}_{t \geq 0}$ as ${(\tilde{P}_t)}_{t \geq 0}$, and for all $k \geq 0$.
%At each time step $t_k = h k$, the resulting discrete-time Markov operator is denoted as $\tilde{K}$ such that $\pi \tilde{P}_{t_k} = \pi \tilde{K}^k$.
For each $k \geq 0$, consider the corresponding time step ${(t_k = h k)}_{k \geq 0}$.
${(Z_t')}_{t \in [t_k, t_{k+1}]}$ is the linear interpolation of ${(Z^*_t)}_{t \geq 0}$ over the interval $[t_k, t_{k+1}]$ with the associated Markov semigroup as ${(\tilde{P}_t)}_{t \in [t_k, t_{k+1}]}$. 
Specifically, for any $t \in [t_k, t_{k+1}]$,
\begin{align}
    V_{t}^{\prime}
    &=
    \mathrm{e}^{- \gamma (t - hk)} V_{hk}^* - \eta \int^t_{hk} \mathrm{e}^{- \gamma \left(s - hk\right)} \nabla U\left(X_{hk}^{*}\right) \mathrm{d}s + \sqrt{2 \gamma \eta} \int^t_{hk} \mathrm{e}^{-\gamma \left(s - hk\right)} \, \mathrm{d}B_s
    \nonumber
    \\
    X_t^{\prime} &= X_{hk}^* + \int^t_{hk} V_s^{\prime} \, \mathrm{d}s \; .
    \label{eq:interpolated_process}
\end{align}
In essence, ${(Z_t')}_{t \in [t_k, t_{k+1}]}$ is a kinetic Langevin diffusion process with the drift set to be the zeroth order interpolation of the drift of ${(Z_{t}^*)}_{t \geq 0}$.
The resulting discrete-time operator is denoted by $\tilde{K}$ such that $\pi \tilde{P}_{t_k} = \pi \tilde{K}^k$.
Given these, we can decompose the asymptotic error as
\[
    \mathrm{W}_{a,b}\left(\pi_h, \pi\right)
    =
    \lim_{k \to \infty} \mathrm{W}_{a,b}\big(\mu K^k, \pi P_{hk}\big)
    \leq
    \lim_{k \to \infty} 
    \left\{
    \mathrm{W}_{a,b}\big(\mu K^k, \pi \tilde{K}^k\big)
    +
    \mathrm{W}_{a,b}\big(\pi \tilde{P}_{t_k}, \pi P_{t_k}\big)
    \right\} \; .
\]
To upper bound this, for all $k \geq 0$ and all $t \in [t_k, t_{k+1}]$, we will assume ${(Z_k)}_{k \geq 0}$, ${(Z_t^*)}_{t \geq 0}$, and ${(Z_t')}_{t \geq 0}$ are synchronously coupled by sharing the same noise process ${(B_t)}_{t \geq 0}$.
Then $\mathrm{W}_{a,b}\big(\mu K^k, \pi \tilde{K}^k\big)$ can be bounded via the synchronous coupling established earlier in~\cref{thm:special_case_convergence_bound}. 
The proof focuses on bounding the remaining $\mathrm{W}_{a,b}\big(\pi \tilde{P}_{t_k}, \pi P_{t_k}\big)$ by analyzing the local error committed at each time step.

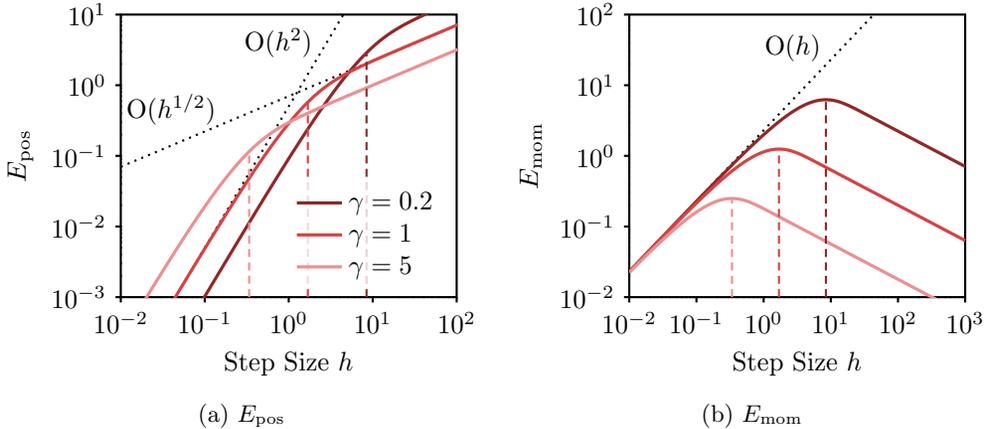
\begin{figure*}
    \centering
    \subfloat[$E_{\mathrm{pos}}$]{
         
\begin{tikzpicture}
    \begin{loglogaxis}[
        xmin=1e-2,
        xmax=1e+2,
        ymin=1e-3,
        ymax=1e+1,
        xlabel={Step Size $h$},
        ylabel={$E_{\mathrm{pos}}$},
        xmode=log,
        ymode=log,
        major tick length=2pt,
        minor tick length=0pt,
        xtick align=outside,
        ytick align=outside,
        ytick pos=left, 
        xtick pos=bottom, 
        axis line style = thick,
        every tick/.style={black,thick},
        width = 40ex,
        height= 35ex,
        legend cell align={left},
        legend style={
            at={(0.97,0.03)},
            anchor=south east,
            draw=none,
            fill opacity=0.8,
            text opacity=1.0,
        },
    ]
    \addplot[black, dotted, thick, domain=1e-2:1e+3, samples=128, forget plot] {0.5*x^(2)} node[black, anchor=south east, pos=0.48, yshift=-1.5ex] {$\mathrm{O}(h^2)$};
    
    \addplot[black, dotted, thick, domain=1e-2:1e+3, samples=128, forget plot] {0.7*sqrt(x)} node[black, anchor=south east, pos=0.25, yshift=-1ex] {$\mathrm{O}(h^{1/2})$};

    \addplot[cherry2, mark=none, very thick, name path=gamma1] table[x index=0, y index=1, col sep=comma] {figures/epos.txt};
    \addlegendentry{$\gamma = 0.2$}
    
    \addplot[cherry4, mark=none, very thick, name path=gamma2] table[x index=0, y index=2, col sep=comma] {figures/epos.txt};
    \addlegendentry{$\gamma = 1$}
    
    \addplot[cherry6, mark=none, very thick, name path=gamma3] table[x index=0, y index=3, col sep=comma] {figures/epos.txt};
    \addlegendentry{$\gamma = 5$}
   
    \path [name path=gamma1thres] (axis cs: 1.69/0.2, 1e-3) -- (axis cs: 1.69/0.2, 1e+2);
    \path [name intersections={of=gamma1 and gamma1thres}];
    \draw[cherry2, thick, densely dashed] (1.69/0.2, 1e-3) -- (intersection-1); 
    
    \path [name path=gamma2thres] (axis cs: 1.69, 1e-3) -- (axis cs: 1.69, 1e+2);
    \path [name intersections={of=gamma2 and gamma2thres}];
    \draw[cherry4, thick, densely dashed] (1.69, 1e-3) -- (intersection-1); 
    
    \path [name path=gamma3thres] (axis cs: 1.69/5, 1e-3) -- (axis cs: 1.69/5, 1e+2);
    \path [name intersections={of=gamma3 and gamma3thres}];
    \draw[cherry6, thick, densely dashed] (1.69/5, 1e-3) -- (intersection-1); 
    
    \end{loglogaxis}
\end{tikzpicture}
    }
    \subfloat[$E_{\mathrm{mom}}$]{
        \begin{tikzpicture}
    \begin{loglogaxis}[
        xmin=1e-2,
        xmax=1e+3,
        ymin=1e-2,
        ymax=1e+2,
        xlabel={Step Size $h$},
        ylabel={$E_{\mathrm{mom}}$},
        xmode=log,
        ymode=log,
        major tick length=2pt,
        minor tick length=0pt,
        xtick align=outside,
        ytick align=outside,
        ytick pos=left, 
        xtick pos=bottom, 
        axis line style = thick,
        every tick/.style={black,thick},
        width = 40ex,
        height= 35ex,
    ]
    
    \addplot[black, dotted, thick, domain=1e-2:1e+3, samples=128] {2.3*x} node[black, anchor=south east, pos=0.6, yshift=-1ex] {$\mathrm{O}(h)$};
    
    \addplot[cherry2, mark=none, very thick, name path=gamma1] table[x index=0, y index=1, col sep=comma] {figures/emom.txt};
    
    \addplot[cherry4, mark=none, very thick, name path=gamma2] table[x index=0, y index=2, col sep=comma] {figures/emom.txt};
    
    \addplot[cherry6, mark=none, very thick, name path=gamma3] table[x index=0, y index=3, col sep=comma] {figures/emom.txt};

    \path [name path=gamma1thres] (axis cs: 1.69/0.2, 1e-2) -- (axis cs: 1.69/0.2, 1e+2);
    \path [name intersections={of=gamma1 and gamma1thres}];
    \draw[cherry2, thick, densely dashed] (1.69/0.2, 1e-2) -- (intersection-1); 
    
    \path [name path=gamma2thres] (axis cs: 1.69, 1e-2) -- (axis cs: 1.69, 1e+2);
    \path [name intersections={of=gamma2 and gamma2thres}];
    \draw[cherry4, thick, densely dashed] (1.69, 1e-2) -- (intersection-1); 
    
    \path [name path=gamma3thres] (axis cs: 1.69/5, 1e-2) -- (axis cs: 1.69/5, 1e+2);
    \path [name intersections={of=gamma3 and gamma3thres}];
    \draw[cherry6, thick, densely dashed] (1.69/5, 1e-2) -- (intersection-1); 
    
    \end{loglogaxis}
\end{tikzpicture}
    }
    \vspace{-1ex}
    \caption{
        \textbf{Scaling of the asymptotic error bound with respect to the step size $h$.}
        The vertical dashed lines mark the point where $\zeta = h\gamma = 1.69$.
    }\label{fig:asymptotic_bias}
    \vspace{-1ex}
\end{figure*}

%
% Meanwhile, a popular alternative strategy~\citep{cheng_underdamped_2018,dalalyan_sampling_2020,chewi_logconcave_2024} is to interpolate the non-stationary Markov operator $K$, for each $k \geq 0$, into a Markov process ${(\hat{P}_t)}_{t \in [hk, h (k + 1)]}$, and rely on the contractivity of the continuous-time dynamics in \cref{eq:kinetic_langevin_dynamics}, while the local error is caused by the interpolation (in contrast to discretization).
% Unfortunately, the non-stationarity of the Markov chain results in terms that are difficult to control.
% As a result, this strategy necessitates stronger constraints on the parameters, resulting in an extra $\alpha$ dependence later in the mixing time complexity~\citep{sanz-serna_wasserstein_2021}.
Specifically, carefully quantifying the effect of the damping (terms involving $\mathrm{e}^{-\gamma t}$) yields the following result.

\begin{theoremEnd}[%
    category=asymptoticbias,
    text link={\noindent\textit{Proof.} The proof is deferred to \cref{section:proof_asymptotic_bias}. \qed},
    text proof = {Proof of \string\pratendRef{thm:prAtEnd\pratendcountercurrent}}
]{theorem}\label{thm:asymptotic_bias}
Suppose \cref{asusmption:hessian_bounded,assumption:condition_linear_approximation_contraction} hold.
Then the stationary distribution of $K$, $\pi_h$, satisfies
\[
    \mathrm{W}_{a,b}\left(\pi_{h}, \pi\right)
    \quad\leq\quad
    E_{\mathrm{pos}} + E_{\mathrm{mom}}
    \; ,
\]
where 
\begin{align*}
    E_{\mathrm{pos}} 
    \quad&\triangleq\quad  
    {\left\{
    \frac{1}{2}
    \frac{d \kappa^2 \eta}{\gamma^2}
    \frac{1}{h \gamma}
    \left(
    {(h\gamma)}^2
    - 
    3
    +
    \delta^2
    \left(
    3
    +
    6 h\gamma
    +
    5 {(h \gamma)}^2
    +
    2
    {(h \gamma)}^3
    \right) 
    \right)
    \right\}}^{1/2}
    \\
    E_{\mathrm{mom}} 
    \quad&\triangleq\quad
    {\left\{
    4
    \frac{d \kappa^2 \eta }{\gamma^2} 
    \frac{1}{h \gamma} 
    \left(
    1
    -
    \delta^2
    \left(
    1 + 2 h \gamma + 2 {(h \gamma)}^2
    \right)
    \right) 
    \right\}}^{1/2} \; .
\end{align*}
\end{theoremEnd}
\begin{proofEnd}
Under the stated assumptions, we can invoke \cref{thm:general_discretization_bound,thm:momentum_deviation_bound}, which yield a bound on the total local error:
\begin{align*}
    \mathrm{W}_{a,b}(\mathrm{Law}(Z'_{h(k+1)}), \mathrm{Law}(Z_{h (k+1)}^*))
    \;\leq\;
    {\left(\mathbb{E} {\lVert Z'_{h(k+1)} - Z_{h (k+1)}^* \rVert}_{a,b}^2\right)}^{1/2}
    \;\leq\;
    \widetilde{E}_{\text{pos}}
    +
    \widetilde{E}_{\text{mom}} \; ,
\end{align*}
where 
\begin{align*}
    \widetilde{E}_{\text{pos}}^2
    \quad&\leq\quad
    \frac{1}{8}
    h d \beta^2 \eta^3
    \Bigg\{
    \frac{h^2}{\gamma^3}
    -
    3
    \frac{1}{\gamma^5}
    +
    \mathrm{e}^{- 2 h \gamma}
    \left(
    3
    \frac{1}{\gamma^5}
    +
    6
    \frac{h}{\gamma^4}
    +
    5
    \frac{h^2}{\gamma^3}
    +
    2
    \frac{h^3}{\gamma^2} 
    \right) 
    \Bigg\}
    \\
    \quad&=\quad
    \frac{1}{8}
    d \beta^2 \eta^3
    \Bigg\{
    \frac{h^3}{\gamma^3}
    -
    3
    \frac{h}{\gamma^5}
    +
    \mathrm{e}^{- 2 h \gamma}
    \left(
    3
    \frac{h}{\gamma^5}
    +
    6
    \frac{h^2}{\gamma^4}
    +
    5
    \frac{h^3}{\gamma^3}
    +
    2
    \frac{h^4}{\gamma^2} 
    \right) 
    \Bigg\}
    \\
    \quad&=\quad
    \frac{1}{8}
    d \beta^2 \eta^3
    \frac{\zeta}{\gamma^6}
    \Bigg\{
    \zeta^2 - 3
    +
    \mathrm{e}^{- 2 \zeta}
    \left(
    3
    +
    6 \zeta
    +
    5 \zeta^2
    +
    2
    \zeta^3
    \right) 
    \Bigg\}
    \\
    \widetilde{E}_{\text{mom}}^2
    \quad&\leq\quad
    \frac{1}{4} 
    a \,
    d \beta^2 \eta^3
    \left\{
    \frac{h}{\gamma^3} 
    -
    \mathrm{e}^{-2 h \gamma }
    \left(
    \frac{h}{\gamma^3} 
    +
    2
    \frac{h^2}{\gamma^2}  
    +
    2
    \frac{h^3}{\gamma}
    \right)
    \right\} 
    \\
    \quad&=\quad
    d \beta^2 \eta^3
    \frac{h}{\gamma^5} 
    \left(
    1
    -
    \mathrm{e}^{-2 h \gamma}
    \left(
    1 + 2 h \gamma + 2 h^2 \gamma^2
    \right)
    \right)
    &&\text{($a = 4/\gamma^2$)}
    \\
    \quad&=\quad
    d \beta^2 \eta^3
    \frac{\zeta}{\gamma^6} 
    \left(
    1
    -
    \mathrm{e}^{-2 \zeta}
    \left(
    1 + 2 \zeta + 2 \zeta^2
    \right)
    \right) 
    \; .
\end{align*}
Denote $\Delta_k \triangleq \mathrm{W}_{a,b}(\mathrm{Law}(Z_{k+1}), \mathrm{Law}(Z_{h (k+1)}^*))$.
Then \cref{eq:bias_partial_contraction_general} yields
\begin{align}
    \Delta_{k+1}
    \quad&\leq\quad
    {\left(1 - \tilde{c}\right)}^{1/2}
    \Delta_k
    +
    \widetilde{E}_{\mathrm{pos}} + \widetilde{E}_{\mathrm{mom}}
    \nonumber
    \; .
\end{align}
Unrolling the recursion, 
\begin{align*}
    \Delta_{k+1}
    \quad\leq\quad
    {\left(1 - \tilde{c}\right)}^{k/2} \Delta_0
    +
    \left(\widetilde{E}_{\mathrm{pos}} + \widetilde{E}_{\mathrm{mom}}\right)
    \sum^{k}_{i = 0} {\left(1 - c\right)}^{i/2}  \; .
\end{align*}
By taking the limit $k \to \infty$, and given the fact that $\lim_{k \to \infty} \mathrm{Law}\left(Z_{k}\right) = \pi_h$ due to \cref{thm:special_case_convergence_bound}, and that ${(Z_t^*)}$ is stationary for all $t \geq 0$,
\begin{align*}
    \mathrm{W}_{a,b}\left(\pi_{h}, \pi\right)
    \quad=\quad
    \Delta_{\infty} 
    \quad\leq\quad
    \left(\widetilde{E}_{\mathrm{pos}} + \widetilde{E}_{\mathrm{mom}}\right)
    \sum^{\infty}_{i = 0} {(1 - c)}^{i/2}
    . 
\end{align*}
Using the bound
\begin{align*}
    \sum^{\infty}_{i = 0} {(1 - \tilde{c})}^{i/2}
    \quad\leq\quad
    \frac{1}{1 - \sqrt{1 - \tilde{c}}}
    \quad\leq\quad
    \frac{2}{\tilde{c}}
    \quad=\quad
    \frac{2 \gamma^2}{\eta \zeta \alpha}
    \; ,
\end{align*}
we finally conclude that
\begin{align*}
    \mathrm{W}_{a,b}\left(\pi_{h}, \pi\right)
    \quad\leq\quad
    \frac{2\gamma^2}{\eta \zeta \alpha}
    \widetilde{E}_{\text{pos}} 
    +
    \frac{2\gamma^2}{\eta \zeta \alpha}
    \widetilde{E}_{\text{mom}} 
    \quad=\quad
    E_{\text{pos}} 
    +
    E_{\text{mom}} 
    \; .
\end{align*}

\end{proofEnd}

Notice that we decomposed the asymptotic error into the local error of the position $E_{\mathrm{pos}}$ and the momentum $E_{\mathrm{mom}}$.
Also, since we rely on \cref{thm:special_case_convergence_bound}, the restriction on the parameters is again given by \cref{assumption:condition_linear_approximation_contraction}.
Although the bounds on $E_{\mathrm{pos}}$ and $E_{\mathrm{mom}}$ are highly non-linear in $\zeta = h \gamma$, making them difficult to interpret, we visualize the scaling with respect to $h$ in \cref{fig:asymptotic_bias} for different values of $\gamma$.
In the underdamped regime, the error of the momentum $E_{\mathrm{mom}}$, which scales as $\mathrm{O}(h)$, dominates the error.
On the other hand, in the overdamped regime, the error of the momentum $E_{\mathrm{mom}}$ \textit{decreases}, while the error of the position $E_{\mathrm{pos}}$ dominates with a $\mathrm{O}(h^{1/2})$ scaling.
This $\mathrm{O}(h^{1/2})$ scaling is typical of the Euler--Maruyama discretization of the overdamped Langevin~\citep[Cor. 7]{durmus_highdimensional_2019}.

Notice in \cref{fig:asymptotic_bias} that there is a phase transition from the underdamped to overdamped regime.
We can locate the critical point of the phase transition by solving for the root of
\[
    \frac{\mathrm{d} E_{\mathrm{mom}}}{\mathrm{d} h}
    =
    0
    \qquad\Leftrightarrow\qquad
    \frac{\mathrm{d} E_{\mathrm{mom}}}{\mathrm{d} \zeta}
    =
    0
    \qquad\Leftrightarrow\qquad
    \left(2 \zeta + 1\right) \left(2 \zeta^2 + 1\right)
    =
    \mathrm{e}^{2 \zeta} \; .
\]
Numerically solving this equation suggests that the critical point is $\zeta \approx 1.69$.
\cref{fig:asymptotic_bias} qualitatively shows that $\zeta = h\gamma = 1.69$ accurately predicts the phase transition for both $E_{\mathrm{pos}}$ and $E_{\mathrm{mom}}$.

By approximating the bounds with the series expansion at both extremes, $\gamma \to 0$ and $\gamma \to \infty$, we can formalize the scalings as follows:

\begin{theoremEnd}[%
    %restate,
    category=asymptoticbiasspecialcases,
    text link={\noindent\textit{Proof.} The proof is deferred to \cref{section:proof_asymptotic_bias_special_cases}. \qed},
    text proof = {Proof of \string\pratendRef{thm:prAtEnd\pratendcountercurrent}}
]{corollary}\label{thm:asymptotic_bias_special_cases}
    Suppose \cref{asusmption:hessian_bounded,assumption:condition_linear_approximation_contraction} hold.
    Then the following bounds hold simultaneously:
    \begin{enumerate}[label=\textrm{(\roman*)}]
        \item Underdamped regime:
    \begin{align*}
        E_{\mathrm{pos}}
        \;\leq\;
        \frac{4}{15} \,
        d^{1/2} \kappa \, \eta^{1/2} \gamma h^2
        \qquad\text{and}\qquad
        E_{\mathrm{mom}}
        \;\leq\;
        \frac{4}{\sqrt{3}}  \,
        d^{1/2} \kappa \, \eta^{1/2} h \; .
    \end{align*}
        
    \item Overdamped regime:
    \begin{align*}
        E_{\mathrm{pos}} 
        \;\leq\;
        \frac{1}{\sqrt{2}} \,
        d^{1/2} \kappa \, \eta^{1/2} \frac{h^{1/2}}{\gamma^{1/2}} 
        \qquad\text{and}\qquad
        E_{\mathrm{mom}} 
        \;\leq\;
        4 \,
        d^{1/2} \kappa \, \eta^{1/2} \frac{1}{h^{1/2} \gamma^{3/2}} 
        \; .
    \end{align*}
    \end{enumerate}
\end{theoremEnd}
\begin{proofEnd}
    We will begin with (i).
    For $E_{\mathrm{pos}}$, we will prove that that, for any $\zeta > 0$, the bound
    \[
        f_{\mathrm{pos}}\left(\zeta\right) \quad\leq\quad \frac{8}{15} \zeta^5
    \]
    holds.
    Denote the right-hand side as $g_{\mathrm{pos}}^{(\mathrm{i})}\left(\zeta\right) \triangleq \frac{8}{15} \zeta^5$.
    Recall the derivatives of $f_{\mathrm{pos}}$ in \cref{eq:fpos_derivative1,eq:fpos_derivative2,eq:fpos_derivative3}.
    Then the bound on $E_{\mathrm{pos}}$ follow from the facts that 
    \begin{align*}
        \frac{\mathrm{d}^3 f_{\mathrm{pos}}}{\mathrm{d}\zeta^3}\left(\zeta\right)
        <
        32 \zeta^2 
        =
        \frac{\mathrm{d}^3 g_{\mathrm{pos}}^{(\mathrm{i})}}{\mathrm{d}\zeta^3} \left(\zeta\right)
        \; ,
        \quad
        \frac{\mathrm{d}^2 f_{\mathrm{pos}}}{\mathrm{d}\zeta^2}\left(0\right) 
        \leq
        \frac{\mathrm{d}^2 g_{\mathrm{pos}}^{(\mathrm{i})}}{\mathrm{d}\zeta^2}\left(0\right)
        \; ,
        \quad 
        \frac{\mathrm{d} f_{\mathrm{pos}}}{\mathrm{d}\zeta}\left(0\right) 
        =
        \frac{\mathrm{d} g_{\mathrm{pos}}^{(\mathrm{i})}}{\mathrm{d}\zeta}\left(0\right)  \; .
    \end{align*}
    That is, $f_{\mathrm{pos}}\left(\zeta\right) < g_{\mathrm{pos}}^{(\mathrm{i})}\left(\zeta\right)$ for all $\zeta > 0$.

    For $E_{\mathrm{mom}}$, we will prove that, for any $\zeta > 0$, the bound
    \[
        f_{\mathrm{mom}}\left(\zeta\right)
        \quad\leq\quad
        \frac{4}{3} \zeta^3
    \]
    holds.
    Denote the right-hand side as $g_{\mathrm{mom}}^{(\mathrm{i})}\left(\zeta\right)  \triangleq \frac{4}{3} \zeta^3$.
    Clearly, for all $\zeta > 0$ and \cref{eq:fmom_derivative1}, 
    \[
        \frac{\mathrm{d} f_{\mathrm{mom}}}{ \mathrm{d}\zeta } < \frac{\mathrm{d} g_{\mathrm{mom}}^{(\mathrm{i})}}{ \mathrm{d}\zeta }
        \qquad\text{and}\qquad
        f_{\mathrm{mom}}\left(0\right)
        =
        g_{\mathrm{mom}}^{(\mathrm{i})}\left(0\right) \; .
    \]
    Therefore, for all $\zeta > 0$, $f_{\mathrm{mom}}\left(\zeta\right) < g_{\mathrm{mom}}^{(\mathrm{i})}\left(\zeta\right)$.
    This implies the bound on $E_{\mathrm{mom}}$.

    Let's turn to (ii).
    The bound on $E_{\mathrm{pos}}$ is equivalent to 
    \[
        f_{\mathrm{pos}}\left(\zeta\right) \quad\leq\quad \zeta^2 \; .
    \]
    Denote the right-hand side as $g^{(\mathrm{ii})}(\zeta) = \zeta^2$ and recall \cref{eq:fpos_derivative1}.
    The bound immediately follows from the fact that 
    \[
        \frac{\mathrm{d} f_{\mathrm{pos}}}{\mathrm{d} \zeta}\left(\zeta\right)
        \leq
        2 \zeta
        =
        \frac{\mathrm{d} g^{(\mathrm{ii})}}{\mathrm{d} \zeta} \left(\zeta\right)
        \qquad\text{and}\qquad
        f\left(0\right) = g^{(\mathrm{ii})}\left(0\right) \; .
    \]
    Finally, the bound on $E_{\mathrm{mom}}$ follows by using the fact that $f_{\mathrm{mom}} < 1$.
\end{proofEnd}

This means that, in the underdamped regime, the asymptotic bias $E_{\mathrm{pos}} + E_{\mathrm{mom}}$ can be bounded as $\mathrm{O}(d^{1/2} \kappa \eta^{1/2} ( h^2 \gamma + h ) )$, while in the overdamped regime, it can be bounded as $\mathrm{O}(d^{1/2} \kappa \eta^{1/2} h^{1/2}\gamma^{-1/2})$.
While both bounds hold simultaneously, they are only tight (with regard to the general bound in \cref{thm:asymptotic_bias}) in their respective regime.
What is worth noting is that, in the underdamped regime, the bias increases with $\gamma$, whereas in the overdamped regime, the bias decreases with $\gamma$.
The fact that bias decreases with $\gamma$ in the overdamped regime is unsurprising since lower friction results in stiffer dynamics.
However, the fact that the bias in the position increases with $\gamma$ in the underdamped regime is counterintuitive.

Compared to previous works, for the underdamped regime with the choice of $\eta = 1$ and $\gamma \asymp \sqrt{\beta}$, our bound is in agreement (\citealp[Thm. 5.3.9]{chewi_logconcave_2026}; \citealp[Thm. 2]{dalalyan_sampling_2020}), as well as with $\gamma \asymp 1$ and $\eta \asymp 1/\beta$~\citep[Eq. 6.1]{sanz-serna_wasserstein_2021}.
Furthermore, notice that $h$ and $\gamma$ in the bound for the overdamped regime appear as a ratio $h^{1/2}\gamma^{-1/2}$.
This immediately implies that this bound is non-vacuous in the overdamped limit $\gamma \to \infty$ under appropriate step size scaling ($h \propto \gamma$).

\begin{corollary}[Overdamped Limit]\label{thm:asymptotic_bias_overdamped}
    Suppose \cref{asusmption:hessian_bounded} holds and, for some $h_{\mathrm{LMC}} > 0$ satisfying $h_{\mathrm{LMC}} \leq 1/(2\beta)$, the parameters are set as $h = h_{\mathrm{LMC}} \, \gamma$, $\eta = 1$, and $\gamma \to \infty$.
    Then the overdamped limits of $E_{\mathrm{\mathrm{pos}}}$ and $E_{\mathrm{\mathrm{mom}}}$ follow as
    \[
        \lim_{\gamma \to \infty} E_{\mathrm{pos}} = (1/\sqrt{2}) \, d^{1/2} \kappa h_{\mathrm{LMC}}^{1/2} 
        \qquad\text{and}\qquad
        \lim_{\gamma \to \infty} E_{\mathrm{mom}} = 0 
        \; .
    \]
\end{corollary}

This coincides with $\mathrm{O}(d^{1/2} \kappa h_{\mathrm{LMC}}^{1/2})$ scaling of the asymptotic bias of the Euler--Maruyama discretization of overdamped Langevin (LMC) under the synchronous Wasserstein coupling strategy~\citep[Cor. 7]{durmus_highdimensional_2019}.
Under \cref{asusmption:hessian_bounded}, however, the convex analaysis strategy of \citet{durmus_analysis_2019} yields a bound of $\mathrm{O}(d^{1/2} \kappa^{1/2} h_{\mathrm{LMC}}^{1/2})$ for LMC.
The dependence on $\kappa$ in \cref{thm:asymptotic_bias_overdamped} is comparably suboptimal.
The equivalence between overdamped KLMC and LMC thus only holds for results obtained under the synchronous Wasserstein coupling strategy.
While it is possible that applying the convex analysis strategy to kinetic Langevin under the exponential integrator would bridge this gap, resulting in a better dependence on $\kappa$, no such result has been demonstrated.

Meanwhile, all previous strategies~\citep{sanz-serna_wasserstein_2021,cheng_underdamped_2018,dalalyan_sampling_2020} resulted in bounds on the asymptotic error that scale as $\mathrm{O}(h)$ with respect to the step size $h$.
In the overdamped regime, where applying time rescaling requires $h \propto \gamma$, these bounds are vacuous in the limit $\gamma \to \infty$.
The main culprit is that previous analyses ignored the effect of stabilization due to high damping through the use of the bound $\mathrm{e}^{-\gamma h} \leq 1$; as $\gamma \to \infty$, $\mathrm{e}^{-\gamma h}$ becomes close \textcolor{purple}{to} 0, reducing the discretization error.
The general strategy for bounding asymptotic bias by \citet{durmus_asymptotic_2024} also suffers from a similar problem.
For a free variable $n > 0$, their strategy unavoidably results in a bound that increases exponentially in time $t = h n$ with a factor of $\exp\left(t\right) = \exp\left(h n\right)$. (See the definition of $\epsilon(n)$ in Ex. 2; \citealp{durmus_asymptotic_2024}.)
In the underdamped regime, the exponential term can be controlled by upper-bounding $h$.
However, in the overdamped regime, time acceleration $t \to \infty$ makes the term uncontrollable, resulting in a vacuous bound.
This also means that results based on the strategy of~\citet{durmus_asymptotic_2024}, \textit{e.g.}, \citep{monmarche_highdimensional_2021,gouraud_hmc_2025,leimkuhler_contraction_2024}, will be ignorant of the precise dependence on $\gamma$.
Thus, the asymptotic bias analysis of other discretizations, such as OBABO or BAOAB~\citep{leimkuhler_rational_2013}, has potential for improvement.

\subsection{Sampling Complexity}\label{section:complexity}

Lastly, we present a non-asymptotic sampling complexity guarantee.
Since the scaling of the asymptotic bias of KLMC in the overdamped regime is strictly worse than in the underdamped regime, we restrict our interest to the latter.
The proof is a straightforward combination of the results in \cref{section:convergence_stationarity,section:asymptotic_bias}.
For any $\mu \in \mathcal{P}_2\left(\mathbb{R}^{2 d}\right)$,any $n \geq 0$, and the choice of $a = 4/\gamma^2$ and $b = 1/\gamma$, we have the decomposition 
\[
    \mathrm{W}_{a,b}\left(\mu K^n, \pi\right) 
    \leq 
    \underbrace{
    \mathrm{W}_{a,b}\left(\mu K^n, \pi_h\right)
    }_{\text{non-stationarity error}}
    + 
    \underbrace{
    \mathrm{W}_{a,b}\left(\pi_h, \pi\right) \; .
    }_{\text{asymptotic bias error}}
\]
Each term is bounded by invoking \cref{thm:special_case_convergence_bound} and  \cref{thm:asymptotic_bias_special_cases}, respectively.
For any given $\epsilon > 0$, solving for the smallest $n$ that ensures both terms are bounded by $\epsilon/2$ yields a sampling complexity result.

\begin{theoremEnd}[%
    category=complexity,
    text link={\noindent\textit{Proof.} The proof is deferred to \cref{section:proof_complexity}. \qed},
    text proof = {Proof of \string\pratendRef{thm:prAtEnd\pratendcountercurrent}}
]{theorem}\label{thm:complexity}
    Suppose \cref{asusmption:hessian_bounded} holds and suppose there exists some $\gamma$, $\eta$, $h_0$ such that \cref{assumption:condition_linear_approximation_contraction} is satisfied for all $h \in (0, h_0]$.
    Then, for any $\epsilon > 0$, 
    \[
        h
        = 
        \min\left\{  \;
            \sqrt{\frac{5}{4}} \frac{\epsilon^{1/2}}{d^{1/4} \kappa^{1/2} \eta^{1/4} \gamma^{1/2}}, \;
            \frac{1}{4 \sqrt{3}} \frac{\epsilon}{d^{1/2} \kappa \eta^{1/2}}, \;
            h_0
        \; \right\}
    \]
    and any number of iterations of at least
    \begin{align*}
        n
        \geq
        \max\left\{  \;
            \sqrt{5} \, \frac{\gamma^{3/2}  }{  \eta^{3/4}}
            \frac{d^{1/4} \kappa^{1/2}}{\alpha}
            \frac{1}{\epsilon^{1/2} }
            , \;
            8 \sqrt{3} \, \frac{\gamma}{\eta^{1/2}}
            \frac{d^{1/2} \kappa}{\alpha}
            \frac{1}{\epsilon}
            , \;
            h_0 \gamma
        \; \right\}
        \log\left( 3 {\mathrm{W}_{a,b}\left(\mu, \pi_h\right)} \frac{1}{\epsilon} \right) 
    \end{align*}
    guarantees that $\mathrm{W}_{a,b}\left(\mu K^n, \pi\right) \leq \epsilon$.
\end{theoremEnd}
\begin{proofEnd}
    Assuming there exists some $\zeta_0$ such that all $\zeta \leq \zeta_0$ satisfies \cref{assumption:condition_linear_approximation_contraction}, according to \cref{thm:Emom_Epos_condition}, the choice
    \begin{align*}
        & &
        \zeta 
        \quad&=\quad 
        \min\left\{  \;
            \sqrt{\frac{5}{4}} \frac{\epsilon^{1/2} \gamma^{1/2}}{d^{1/4} \kappa^{1/2} \eta^{1/4}}, \;
            \frac{1}{4 \sqrt{3}} \frac{\epsilon \gamma}{d^{1/2} \kappa \eta^{1/2}}, \;
            \zeta_0
        \; \right\}
        \; .
    \end{align*}
    ensure $E_{\mathrm{pos}} \leq \frac{\epsilon}{3} \land E_{\mathrm{mom}} \leq \frac{\epsilon}{3}$.
    For the number of steps $n$ guaranteeing $E_{\mathrm{stat}} \leq \frac{\epsilon}{3}$, it suffices to substitute our choice of $\zeta$ in \cref{thm:Estat_condition}.
    This yields  
    \begin{align*}
        n
        &\geq
        \max\left\{  \;
            \sqrt{5} \, \frac{\gamma^{3/2}  }{  \eta^{3/4}}
            \frac{d^{1/4} \kappa^{1/2}}{\alpha}
            \frac{1}{\epsilon^{1/2} }
            , \;
            8 \sqrt{3} \, \frac{\gamma}{\eta^{1/2}}
            \frac{d^{1/2} \kappa}{\alpha}
            \frac{1}{\epsilon}
            , \;
            \zeta_0
        \; \right\}
        \log\left( 3 {\mathrm{W}_{a,b}\left(\mu, \pi_h\right)} \frac{1}{\epsilon} \right) 
        \; .
    \end{align*}
\end{proofEnd}

Substituting $\gamma$, $\eta$, and $h_0$ with values that satisfy \cref{assumption:condition_linear_approximation_contraction} yields a more concrete complexity guarantee.
In particular, we retrieve previous sampling complexity guarantees for the KLMC in the underdamped regime.

\begin{corollary}\label{thm:explicit_iteration_complexity}
    For any $\epsilon > 0$ and the parameters $h$, $\gamma$, $\eta$ satisfying any of the following choices:
    \begin{enumerate}
        \item $\gamma = \sqrt{27 \beta}$, $\eta = 1$, $h = \mathrm{O}\left(\epsilon/\left(d^{1/2} \kappa \eta^{1/2}\right)\right)$
        
        \item $\gamma = \sqrt{27/2}$, $\eta = 1/(2 \beta)$, $h = \mathrm{O}\left(\alpha^{1/2} \epsilon / \left( d^{1/2} \kappa^{1/2}\right) \right)$
    \end{enumerate}
    and a number of iterations of at least
    \[
        n \geq \mathrm{O}\left( d^{1/2} \kappa^{3/2} \frac{1}{ \alpha^{1/2} \epsilon}  \log \frac{1}{ \alpha^{1/2} \epsilon}  \right) \; , 
    \]
    we have that $\mathrm{W}_{a,b}\left(\mu K^n, \pi\right) \leq \epsilon$.
\end{corollary}

For the dimensionless target condition $ \mathrm{W}_{a,b}\left(\mu K^n, \pi\right) \leq \alpha^{-1/2} \epsilon$, \cref{thm:explicit_iteration_complexity} provides an iteration complexity of $\mathrm{O}\left( d^{1/2} \kappa^{3/2} \epsilon^{-1} \log \epsilon^{-1} \right)$, which matches known results~\citep{sanz-serna_wasserstein_2021,dalalyan_sampling_2020,altschuler_shifted_2025}.

\section{Discussions}\label{section:discussions}
In this work, we have presented new results on the Wasserstein contraction (\cref{section:convergence_stationarity}) and asymptotic bias (\cref{section:asymptotic_bias}) of kinetic Langevin dynamics discretized via the stochastic exponential Euler scheme. 
% The technique conssits 
% through the classical Wasserstein synchronous coupling technique.
Our results are in accordance with past results~\citep{sanz-serna_wasserstein_2021,dalalyan_sampling_2020,altschuler_shifted_2025} in the underdamped regime, but are general enough to accommodate all of the different parameter choices in past works,
for example, large friction $\gamma \asymp \sqrt{\beta}$ with fixed inverse mass $\eta \asymp 1$~\citep{dalalyan_sampling_2020,altschuler_shifted_2025}
 versus fixed friction $\gamma \asymp 1$ with small inverse mass $\eta \asymp 1/\beta$~\citep{sanz-serna_wasserstein_2021}.
Furthermore, we have extended the convergence guarantees of KLMC to the overdamped regime ($\gamma \to \infty$).
This demonstrates that the exponential integrator does \textit{not} degenerate in the overdamped limit as long as proper time scaling $h \propto \gamma$ is applied.
This contrasts with the conclusion of \citet{leimkuhler_contraction_2024}, as their analysis required stronger constraints on the parameters, which prevented time acceleration.
While our results are primarily of theoretical interest, they yield a practical conclusion: For a given $\gamma$, larger step sizes ($h \propto \gamma$) can be used than previously prescribed by the theory ($h \propto 1/\gamma$).
One limitation of our analysis, however, is that the benefits of larger step sizes did not materialize in the complexity analysis (\cref{thm:explicit_iteration_complexity}).
This is because the asymptotic bias bound in \cref{thm:asymptotic_bias} is too complicated to be used in the complexity analysis, and the approximation in \cref{thm:asymptotic_bias_special_cases} had to be used.

Now that we know the fact that the exponential integrator does not degenerate as long as proper time acceleration is applied, it is rather surprising that the OBABO and BAOAB discretization of the kinetic Langevin diffusion reduces to LMC \textit{without} any sort of explicit time scaling of the parameters~\citep{leimkuhler_contraction_2024}.
This suggests that, somehow, splitting schemes implicitly experience automatic time acceleration, which would be interesting to identify in their respective analysis.
Another related future direction, which was raised by one of the reviewers during review, is whether we could systematically check the non-degeneracy of discretization schemes.
\textit{i.e.} come up with a general condition that tests whether a specific discretization scheme achieves a non-degenerate overdamped limit.
As this work focuses specifically on the exponential integration scheme, we leave these to future work.

Furthermore, our bounds on asymptotic bias revealed a more precise dependence on the friction $\gamma$.
In particular, the bias increases with $\gamma$ in the underdamped regime, whereas it decreases with $\gamma$ in the overdamped regime.
It would be interesting to see if a similar dependence on $\gamma$ exhibits in the asymptotic bias of alternative discretization schemes.
Furthermore, our bounds on the asymptotic bias exhibit a clear phase transition from the overdamped regime to the underdamped regime.
Numerically solving the critical point suggests that the transition happens around $h \gamma = 1.69$.
% This has a clear practical implication.
% Currently, the best results on the sampling complexity of KLMC~\citep{sanz-serna_wasserstein_2021,dalalyan_sampling_2020,altschuler_shifted_2025} imply that its dependence on the condition number $\kappa$ is worse than LMC~\citep{durmus_analysis_2019}.
% Therefore, as pointed out by~\citet{dalalyan_sampling_2020}, depending on the target $\epsilon$ and $\kappa$, KLMC may or may not be more efficient than LMC.
% In practice, especially under the complication of parameter tuning, determining when to use KLMC or LMC is not straightforward.
% Our results suggest that one approach is to first tune the parameters of KLMC, and monitor the value $\zeta = h \gamma$.
% If $\zeta$ is much larger than $1.96$, this means KLMC is operating in a regime close to LMC, meaning that it would be better to switch to the computationally cheaper LMC or stay with KLMC otherwise.
It would be curious to see if this phase transition also consistently happens at $h \gamma \approx 1.69$ for other discretizations as well.

\newpage
\section{Proofs}\label{section:proofs}
\subsection{Wasserstein Contraction Analysis}\label{section:proof_general_contraction}

This section will present the proof for \cref{thm:general_contraction} as well as the necessary components for the proof.
The proof follows a synchronous Wasserstein coupling strategy similarly to that of \citet{leimkuhler_contraction_2024}, and can be viewed as a refinement of their strategy.
Consider two Markov chains ${(Z_k = (X_k, V_k))}_{k \geq 0}$ and ${(Z_k' = (X_k', V_k'))}_{k \geq 0}$ following the update in \cref{eq:klmc_zeta_update} with a shared noise process ${(\xi^X_k, \xi^V_k)}_{k \geq 1}$.
That is, for each $k \geq 0$, 
\begin{align*}
    \begin{cases}
	X_{k+1} &= X_k + \left( 1 - \delta \right) \left(\frac{1}{\gamma} V_k \right) - \left(\zeta + \delta - 1\right) \left(\frac{\eta }{\gamma^2} \nabla U\left(X_k\right) \right) + \xi^X_{k+1}
	\\
	V_{k+1} &= \delta V_k - \left(1 - \delta\right) \left( \frac{\eta}{\gamma} \nabla U\left(X_k\right) \right) + \xi^V_{k+1}
    \end{cases}
    \\
    \begin{cases}
	X_{k+1}^{\prime} &= X_k^{\prime} + \left( 1 - \delta \right) \left(\frac{1}{\gamma} V_k^{\prime} \right) - \left(\zeta + \delta - 1\right) \left(\frac{\eta }{\gamma^2} \nabla U\left(X_k^{\prime}\right) \right) + \xi^X_{k+1}
	\\
	V_{k+1}^{\prime} &= \delta V_k^{\prime} - \left(1 - \delta\right) \left( \frac{\eta}{\gamma} \nabla U\left(X_k^{\prime}\right) \right) + \xi^V_{k+1} \; .
    \end{cases}
\end{align*}
Since the two processes ${(Z_k)}_{k \geq 0}$ and ${(Z_k^{\prime})}_{k \geq 0}$ are sharing the same noise process ${(\xi^X_k, \xi^V_k)}_{k \geq 1}$, they are \textit{synchronously  coupled}.
The proof is dedicated to establishing that, for all $k \geq 0$, a contraction holds for some fixed $c > 0$ holds for the norm $\norm{\cdot}_{a,b}$ introduced in \cref{section:norm} as
\begin{equation}
    \norm{Z_{k+1} - Z_{k+1}'}_{a,b}^2 \leq \left(1 - c\right) \norm{Z_k - Z_k'}_{a,b}^2 \; .
    \label{eq:one_step_contraction}
\end{equation}

As the two Markov chains ${(Z_k)}_{k \geq 0}$ and ${(Z_k')}_{k \geq 0}$ are synchronously coupled, conditional on the noise sequence ${(B_k)}_{k \geq 0}$ and a corresponding linear operator $S_k$, the difference sequence ${(\bar{Z}_k = Z_k - Z_k')}_{k \geq 0}$ is a deterministic time-variant linear dynamical system
\[
    \bar{Z}_{k+1} \quad=\quad S_k \bar{Z}_k \; .
\]
Note that $S_k$ is dependent on the states $Z_k$ and $Z_k'$, and the parameters $h$, $\eta$, and $\gamma$.
The contraction in \cref{eq:one_step_contraction} follows from a Lyapunov analysis of the operator $S_k$.
Consider the quadratic Lyapunov function 
\[
    z \mapsto \norm{z}_{a,b}^2 = z^{\top} G z \; , \quad\text{where}\quad G \triangleq \begin{bmatrix}
        1 & b \\
        b & a \\
    \end{bmatrix} \; .
\]
Due to the relationship 
\begin{align*}
    \norm{\bar{Z}_{k+1}}_{a,b}^2 
    \quad=\quad
    \norm{S_k \bar{Z}_{k}}_{a,b}^2 
    \quad=\quad
    \bar{Z}_{k}^{\top} S_k^{\top} G S_k \bar{Z}_{k} \; , 
\end{align*}
the existence of $c > 0$ such that \cref{eq:one_step_contraction} holds can be reduced to solving a special case of the ``discrete-time Lyapunov equation''~\citep[\S 6.E]{antsaklis_linear_2006}
\begin{equation}
    S_k^{\top} G S_k - G \;\preceq\; - c G \; .
    \label{eq:lyapunov_equation}
\end{equation}
In particular, the feasibility of this matrix inequality implies a strong form of Lyapunov stability known as exponential stability.
(See also Exercise 6.10 by~\citet{antsaklis_linear_2006}, which presents the continuous-time analog of this statement.)
This is also the condition identified by~\citet[Prop. 4.6]{sanz-serna_wasserstein_2021}, although they didn't draw connections with Lyapunov analysis of linear systems.

In our case, establishing \cref{eq:lyapunov_equation} corresponds to checking the positive-definiteness of a $2 \times 2$ block matrix comprised of the blocks $A_k, B_k, C_k \in \mathbb{R}^{d \times d}$,
\begin{align}
    \left(1 - c\right) G - S_k^{\top} G S_k
    \quad\triangleq\quad
    \begin{bmatrix}
        A_k & B_k \\
        B_k & C_k
    \end{bmatrix}  \; .
    \label{eq:block_matrix}
\end{align}
While \citet{sanz-serna_wasserstein_2021} avoided directly analyzing the spectrum of this matrix, \citet{leimkuhler_contraction_2024} relied on the following lemma to enable a direct analysis:
\begin{lemma}\label{thm:lyapunov_equation_block_analysis}
    Suppose the block matrices in \cref{eq:block_matrix} are symmetric and commutative.
    Then, there exists a constant $c > 0$ satisfying \cref{eq:lyapunov_equation} for all $k \geq 0$ if and only if $A_k \succ 0$ and $A_k C_k - B_k^2 \succeq 0$ for all $k \geq 0$.
\end{lemma}
\begin{proof}
    Any block matrix of the form of \cref{eq:block_matrix} is positive semidefinite if and only if $A_k \succ 0$ and its Schur complement is positive semidefinite such that $C_k - B_k^{\top} A_k^{-1} B_k \succeq 0$ \citep[\S A.5.5]{boyd_convex_2004}.
    Since the blocks $A_k, B_k, C_k$ are symmetric and commutative, 
    \[
        C_k - B_k^{\top} A_k^{-1} B_k \succeq 0
        \qquad\Leftrightarrow\qquad
        C_k - A_k^{-1} B_k^{\top} B_k \succeq 0 \; .
    \]
    Multiplying $A \succ 0$ on the left of both hand sides, 
    \[
        C_k - A_k^{-1} B_k^{\top} B_k \succeq 0
        \qquad\Leftrightarrow\qquad
        A_k C_k - B_k^2 \succeq 0 \; ,
    \]
    which is the stated result.
\end{proof}
Therefore, it is sufficient to show that, for any $k \geq 0$ and some $h, \eta, \gamma, c > 0$, the conditions 
\begin{equation}
    A_k \succ 0 \quad\text{and}\quad A_k C_k - B_k^2 \succeq 0 
    \label{eq:block_spectrum_condition}
\end{equation}
hold.
The derivation of the blocks is presented in \cref{section:derivation_blocks}, where the proof of \cref{thm:general_contraction}, which consists of establishing \cref{eq:block_spectrum_condition}, proceeds in \cref{section:proof_general_contraction_proof}.

\subsubsection{Derivation of $A_k$, $B_k$, and $C_k$}\label{section:derivation_blocks}
In this section, we compute the blocks $A_k$, $B_k$, and $C_k$.
Since $U$ is twice differentiable under \cref{asusmption:hessian_bounded}, we can invoke the fundamental theorem of calculus, which, for all \(x, x' \in \mathbb{R}^d\), yields the identity
\[
    \nabla U\left(x\right) - \nabla U\left(x^{\prime}\right) = H \left(x - x^{\prime}\right)
    \quad\text{with}\quad
    H\left(x, x'\right) \triangleq \int^1_0 \nabla^2 U\left(x + t \left(x^{\prime} - x\right) \right) \mathrm{d}t \; .
\]
Define $H_k \triangleq H(X_k, X_k')$.
Then the difference sequence satisfies
\begin{align*}
    \bar{Z}_{k+1}
    &=
    \begin{bmatrix}
	  \left( X_k - X^{\prime}_k \right) + \frac{1 - \delta}{\gamma} \left(V_k - V^{\prime}_k\right) - \eta \frac{\zeta + \delta - 1}{\gamma^2} \left( \nabla U\left(X_k\right) - \nabla U\left(X_k^{\prime}\right) \right)
        \\
        \delta \left(V_k - V_k^{\prime}\right) - \eta \frac{1 - \delta}{\gamma} \left( \nabla U\left(X_k\right) - \nabla U\left(X_k^{\prime}\right) \right)
    \end{bmatrix}
    \\
    &=
    \begin{bmatrix}
	\left( X_k - X^{\prime}_k \right) + \frac{1 - \delta}{\gamma} \left(V_k - V^{\prime}_k\right) - \eta \frac{\zeta + \delta - 1}{\gamma^2} H_k \left( X_k - X_k^{\prime} \right)
	\\
	\delta \left(V_k - V_k^{\prime}\right) - \eta \frac{1 - \delta}{\gamma} H_k \left( X_k - X_k^{\prime} \right)
    \end{bmatrix}
    \\
    &=
    \begin{bmatrix}
	  \mathrm{I}_d - \eta \frac{\zeta + \delta - 1}{\gamma^2} H_k & \frac{1 - \delta}{\gamma} \; \mathrm{I}_d
        \\
        - \eta \frac{1 - \delta}{\gamma} H_k & \delta \; \mathrm{I}_d
    \end{bmatrix} 
    \bar{Z}_k \; .
\end{align*}
Therefore, the time-variant transition operator of the difference sequence is
\[
    S_k = \begin{bmatrix}
    \mathrm{I}_d - \eta \frac{\zeta + \delta - 1}{\gamma^2} H_k & \frac{1 - \delta}{\gamma} \; \mathrm{I}_d
    \\
    - \eta \frac{1 - \delta}{\gamma} H_k & \delta \; \mathrm{I}_d 
    \end{bmatrix}.
\]

Since $H_k$ is symmetric, it is diagonalizable for all $k \geq 0$.
Furthermore, all blocks in $S_k$ only involve $H_k$ and $\mathrm{I}_d$, meaning that they are all diagonalizable with the same eigenvectors, which also implies that the blocks commute.
Using this fact and substituting $a = 4/\gamma^2$ and $b = 1/\gamma$ in $G$, symbolic computation shows that the matrices $A_k$, $B_k$, and $C_k$ are symmetric and follow as
\begin{align*}
    A_k
    &=
    \left(- \zeta^2 - 3 {\left( \delta - 1 \right)}^2 \right) {\left( \frac{\eta}{\gamma^2} H_k \right)}^2
    + 
    2 \gamma h \left( \frac{\eta}{\gamma^2}  H_k \right) - c \mathrm{I}_d 
    \\
    B_k
    &=
    \frac{1}{\gamma}
    \left(
        \zeta - 3 \delta^2 + 3 \delta
    \right)
    \left( \frac{\eta}{\gamma^2} H_k \right)  - \frac{1}{\gamma} c \mathrm{I}_d
    \\
    C_k
    &=
    -\frac{1}{\gamma^{2}} \left(4 c - 3 \left(1 - \delta^{2} \right)\right) \mathrm{I}_d \; .
\end{align*}
(The material for replicating the symbolic computation results is available as supplementary material. See \cref{section:supplementary} for details.)

Furthermore, again from symbolic computation, $A_k C_k - B_k^2$ follows as
\begin{align*}
    A_k C_k - B_k^2
    &=
    \frac{1}{\gamma^2}
    \Bigg\{
    \left(12 c {\left(1 - \delta\right)}^2 - 4 \zeta^{2} \left(1 - c\right) + 3 \zeta^{2} \delta^{2} - 6 \zeta \left(\delta - \delta^{2} \right) - 9 {\left(1 - \delta\right)}^2 \right) { \left( \frac{\eta}{\gamma^2} H_k \right) }^2 
    \\
    &\qquad\qquad
    + 6 \left(c \left(\delta - \delta^{2}\right) + \zeta \left(1 - c \right) - \zeta \delta^{2}\right) \left( \frac{\eta}{\gamma^2} H_k \right)
    + \left(3 c^{2} - 3 c \left(1 - \delta^{2}\right)\right)\mathrm{I}_d
    \Bigg\} \; .
\end{align*}
We will now proceed to establish the positive-definiteness of $A_k$ and $A_{k} C_k - B_k^2$.

\subsubsection{Proof of \cref{thm:general_contraction}}\label{section:proof_general_contraction_proof}

Under \cref{asusmption:hessian_bounded}, all of the eigenvalues of $H_k$ are strictly positive.
This implies that the $p$th eigenvalue of $A_k$ and $A_k C_k - B_k^2$ follow as
\begin{align*}
    &
    \sigma_{p}\left(A_k\right)
    =
    \left(- \zeta^2 - 3 {\left( \delta - 1 \right)}^2 \right) {\left( \frac{\eta \sigma_p\left(H_k\right)}{\gamma^2} \right)}^2
    + 
    2 \zeta {\left( \frac{\eta \sigma_p\left(H_k\right)}{\gamma^2} \right)} - c 
    \\
    &
    \sigma_{p}\left(A_k C_k - B_k^2\right)
    =
    \\
    & \qquad
    \frac{1}{\gamma^2}
    \Bigg\{
    \left(12 c {\left(1 - \delta\right)}^2 - 4 \zeta^{2} \left(1 - c\right) + 3 \zeta^{2} \delta^{2} - 6 \zeta \left(\delta - \delta^{2} \right) - 9 {\left(1 - \delta\right)}^2 \right) { \left( \frac{\eta \sigma_p\left(H_k\right)}{\gamma^2}  \right) }^2 
    \\
    &\qquad\qquad\qquad\qquad 
    + 6 \left(c \left(\delta - \delta^{2}\right) + \zeta \left(1 - c \right) - \zeta \delta^{2}\right) \left( \frac{\eta \sigma_p\left(H_k\right)}{\gamma^2} \right)
    + 3 c^{2} - 3 c \left(1 - \delta^{2}\right)
    \Bigg\}  
    \; .
\end{align*}
For notational convenience, consider the functions
\begin{align*}
    \chi_A\left(r, \zeta, \gamma, c\right)
    &\triangleq
    \left(- \zeta^2 - 3 {\left( 1 - \delta \right)}^2 \right) r^2
    + 
    2 \zeta r - c 
    \\
    \chi_{A C - B^2}\left(r, \zeta, \gamma, c\right)
    &\triangleq
    \frac{1}{\gamma^2}
    \Bigg\{
    \left(12 c {\left(1 - \delta\right)}^2 - 4 \zeta^{2} \left(1 - c\right) + 3 \zeta^{2} \delta^{2} - 6 \zeta \left(\delta - \delta^{2} \right) - 9 {\left(1 - \delta\right)}^2 \right) r^2
    \\
    &\qquad\qquad
    + 6 \left(c \left(\delta - \delta^{2}\right) + \zeta \left(1 - c \right) - \zeta \delta^{2}\right) r
    + 3 c^{2} - 3 c \left(1 - \delta^{2}\right)
    \Bigg\}
    \; .
\end{align*}
These functions characterize the spectrum of $A_k$ and $A_k C_k - B_k^2$ via the relationship
\begin{equation}
    \begin{split}
    \chi_A\left( R\left(\sigma_p(H_k)\right), \zeta, \gamma, c\right)
    &=
    \sigma_p\left(A_k\right)  \; ,
    \\
    \chi_{AC-B^2}\left( R\left(\sigma_p(H_k)\right), \zeta, \gamma, c\right)
    &=
    \sigma_p\left(A_k C_k-B_k^2\right) \; .
    \label{eq:chi_spectrum_equivalence}
    \end{split}
\end{equation}
Therefore, analyzing $\chi_A$ and $\chi_{AC-B^2}$ sufficiently characterizes the spectrum of $A_k$ and $A_k C_k - B_k^2$, respectively.
Specifically, we are interested in the conditions on $h, \gamma, r, c$ that guarantees $\chi_A > 0$ and $\chi_{AC-B^2} \geq 0$.
For this, we need to analyze $\chi_A$ and $\chi_{AC-B^2}$ in detail.

% under \cref{asusmption:hessian_bounded}, for any fixed $\zeta, \eta, \gamma, c > 0$, and all $\lambda \in \left[ \alpha, \beta \right]$, 
% \begin{align*}
%     &\forall \lambda \in [\alpha, \beta], \;\;
%     \chi_A\left(R(\lambda), \zeta, \gamma, c\right) > 0
%     \;\;\text{and}\;\;
%     \chi_{AC-B^2}\left(R(\lambda), \zeta, \gamma, c\right) \geq 0
%     \\
%     \quad&\Rightarrow\quad
%     \forall k \geq 0, \;\;
%     A_k \succ 0 
%     \;\;\text{and}\;\;
%     A_k C_k-B^2 \succeq 0
%     \; .
% \end{align*}
% The remainder of the proof focuses on establishing sufficient conditions on $h, \gamma, \eta, c > 0$.

Let's begin with $\chi_{AC - B^2}$.
Notice that it can be rewritten as
\begin{align*}
    &
    \chi_{A C - B^2}\left(r, \zeta, \gamma, c\right)
    \\
    &
    \qquad\triangleq 
    \frac{1}{\gamma^{2}}
    \Big\{
    3 c^{2} + \left(4 \left(3 {\left( 1 - \delta \right)}^2 + \zeta^{2}\right) r^{2}  
    + 6 \left(\left( \delta - \delta^{2} \right) - \zeta\right) r - 3 \left(1 - \delta^{2} \right) \right) c
    \\
    & &
    \mathllap{
    +
    \left(- 9 {\left( 1 - \delta \right)}^2 + \zeta^{2} \left(3 \delta^{2} - 4\right) - 6 \zeta \left( \delta - \delta^{2} \right)\right) r^{2} + 6 \zeta \left(1 - \delta^{2}\right) r
    \Big\}
    }
    \; .
\end{align*}
It is apparent that this is a quadratic function with respect to \(c\).
Furthermore, this quadratic always has two distinct real roots:

\begin{theoremEnd}[%
    restate,
    text proof={},
    category=ACmB2distinctroots,
    text link={\noindent\textit{Proof}. The proof is deferred to \cref{section:proof_ACmB2_distinct_roots}.\qed}
]{lemma}\label{thm:ACmB2_distinct_roots}
    For any \(r, \zeta, \gamma > 0\), the equation \( \chi_{AC - B^2}\left(r, \zeta, \gamma, c\right) = 0 \)
    has two distinct real roots with respect to \(c \in \mathbb{R}\).
\end{theoremEnd}
\begin{proofEnd}
The discriminant of \(c \mapsto \chi_{A C - B^2}\left(r, \zeta, \gamma, c\right) \) is given as
\[
    \mathrm{disc}\left( c \mapsto \chi_{A C - B^2}\left(r, \zeta, \gamma, c\right)
     \right) 
    = 
    p_2\left(r\right) p_3\left(r\right) \; ,
\]
where both $p_2$ and $p_3$ are quadratics in $r$.
Since the leading coefficient of $p_2$ and $p_3$ are \(a_1 > 0\), both are convex.
Now, if a convex quadratic in \(r \in \mathbb{R}\) has a strictly negative discriminant, it is strictly positive for all values of \(r \in \mathbb{R}\).

The discriminant of $p_2$ can be shown to be strictly negative since
\begin{align*}
    \mathrm{disc}\left(p_2\right)
    &=
    b_2^2 - 4 a_1 e_2
    \\
    &=
    {\left( \zeta \left(1 + \delta\right) - {\left(1 - \delta^2\right)} \right)}^2
    -
    4 \left( \frac{2}{3} \zeta^{2} + 2 \left(1 - \delta\right)^{2} \right) \left( \frac{1}{2} {\left(1 + \delta\right)}^2 \right)
    \\
    &=
    {\left(1 - \delta\right)}^2
    \left\{
        {\left( \zeta \, \frac{1 + \delta}{1 - \delta} - \left(1 + \delta\right) \right)}^2
        -
        2 \left( \frac{2}{3} \zeta^{2} + 2 \left(1 - \delta\right)^{2}  \right) {\left( \frac{ 1 + \delta }{ 1 - \delta } \right)}^2
    \right\}
    \\
    &=
    {\left(1 - \delta\right)}^2
    \left\{
        \zeta^2 {\left( \frac{ 1 + \delta }{ 1 - \delta } \right)}^2 - 2 \zeta \frac{{\left( 1 + \delta \right)}^2}{1 - \delta} + {\left( 1 + \delta \right)}^2
        -
        \frac{4}{3} \zeta^{2} {\left( \frac{ 1 + \delta }{ 1 - \delta } \right)}^2 - 4 \left(1 + \delta\right)^{2} 
    \right\}
    \\
    &=
    {\left(1 - \delta\right)}^2
    \left\{
        -\frac{1}{3} \zeta^2 {\left( \frac{ 1 + \delta }{ 1 - \delta } \right)}^2
        - 2 \zeta \frac{ {\left( 1 + \delta \right)}^2 }{1 - \delta}
        - 3 \left(1 + \delta\right)^{2}
    \right\}
    \\
    &< 0 \; ,
\end{align*}
while that of $p_3$ also turns out to be strictly negative since
\begin{align*}
    \mathrm{disc}\left(p_3\right)
    &=
    b_3^2 - 4 a_1 e_3
    \\
    &=
    {\left( \zeta \left(1 - \delta\right) + {\left(1 - \delta\right)}^2 \right)}^2
    -
    4 \left( \frac{2}{3} \zeta^{2} + 2 \left(1 - \delta\right)^{2} \right) \left( \frac{1}{2} {\left(1 - \delta\right)}^2 \right)
    \\
    &=
    {\left(1 - \delta\right)}^2
    \left\{
        {\left( \zeta + \left(1 - \delta\right) \right)}^2
        -
        2 \left( \frac{2}{3} \zeta^{2} + 2 \left(1 - \delta\right)^{2} \right)
    \right\}
    \\
    &=
    {\left(1 - \delta\right)}^2
    \left\{
        \zeta^2 + 2 \zeta \left(1 - \delta\right) + {\left( 1 - \delta \right)}^2
        -
        \frac{4}{3} \zeta^{2} - 4 \left(1 - \delta\right)^{2}
    \right\}
    \\
    &=
    {\left(1 - \delta\right)}^2
    \left\{
        - \frac{1}{3} \zeta^{2} 
        + 2 \zeta \left( 1 - \delta \right) - 3 {\left( 1 - \delta \right)}^2 
    \right\}
    \\
    &=
    {\left(1 - \delta\right)}^2
    \left\{
        - {\left( 
            \frac{1}{\sqrt{3}} \zeta
            - \sqrt{3} \left( 1 - \delta\right)
        \right)}^2
    \right\}
    \\
    &\;< 0.
\end{align*}
Therefore, $p_1 > 0$ and $p_3 > 0$.
Evidently, this implies that $\mathrm{disc}(c \mapsto \chi_{AC-B^2}\left(r, \zeta, \gamma, c\right))$ is always strictly positive and that $c \mapsto \chi_{AC-B^2}\left(r, \zeta, \gamma, c\right)$ always has two distinct roots.
\end{proofEnd}

By the quadratic formula, the roots can be found via symbolic computation as
\begin{align*}
    c^{\pm}\left(r, \zeta\right)
    \quad=\quad
    p_1\left(r\right) \; \pm \; \sqrt{ p_2\left(r\right) p_3\left(r\right) }
\end{align*}
with the quadratics $p_1$, $p_2$, and $p_3$ defined in the proof statement.
For convenience, let us restate the constants in the $\zeta$ parametrization:
\begin{alignat*}{5}
	a_1 &= \frac{2}{3} \zeta^2 + 2 {\left(1 - \delta\right)}^2 \; ,
	\qquad
	&&b_1 &&= \zeta - \left(\delta - \delta^2\right) \; ,
	\qquad
	&&e_1 &&= \frac{1}{2} \left(1 - \delta^2\right) \; ,
	\qquad
	\\
	&
	\qquad
	&&b_2 &&= - \zeta \left(1 + \delta\right) + \left(1 - \delta^2\right) \; ,
	\qquad
	&&e_2 &&= \frac{1}{2} {\left(1 + \delta\right)}^2 \; ,
	\\
	&
	\qquad
	&&b_3 &&= -\zeta \left(1 - \delta\right) - {\left(1 - \delta\right)}^2 \; ,
	\qquad
	&&e_3 &&= \frac{1}{2} {\left(1 - \delta\right)}^2 \; .
\end{alignat*}

Since the leading coefficient of $c \mapsto \chi_{AC-B^2}(r, \zeta, \gamma, c)$ is positive, $c \mapsto \chi_{AC-B^2}(r, \zeta, \gamma, c)$ is a convex quadratic. 
Therefore, its left root
\begin{align}
    c^-(r, \zeta) \quad=\quad 
    p_1\left(r\right) - \sqrt{ p_2\left(r\right) p_3\left(r\right) } \;
    \label{eq:cminus}
\end{align}
identifies the region where $\chi_{AC-B^2}$ is positive.

\begin{lemma}\label{thm:ACmB2_nonnegative}
    For any $r, \zeta, \gamma > 0$ and all $c \in (-\infty, c^-\left(r, \zeta\right)]$, $\chi_{AC-B^2}\left(r, \zeta, \gamma, c\right) \geq 0$.
\end{lemma}
\begin{proof}
    The result follows from the fact that $c^-$ is the left root of the equation $\chi_{AC-B^2}\left(r, \zeta, \gamma, c\right) = 0$ with respect to $c$ and that $\chi_{AC-B^2}$ is a convex quadratic with respect to $c$.
\end{proof}

However, recall that the argument $c$ must satisfy $c > 0$ to be a valid contraction coefficient.
Therefore, we must identify the conditions on $r, \zeta > 0$ that leads to $c^- > 0$ such that the range $(0, c^-)$ is non-empty.
Define
\begin{equation}
    r_{\mathrm{max}}(\zeta)
    \quad\triangleq\quad
    \frac{
        2 \zeta  \left(1 - \delta^2\right) 
    }{
        \left(4/3 - \delta^2\right) \zeta^2  + 2 \delta \left(1 - \delta\right) \zeta + 3 \left(1 - \delta\right)^2
    }  \; .
    \label{eq:rmax}
\end{equation}
In the next lemma, we show that, for any fixed $\zeta > 0$, we have $c^-\left(r, \zeta\right) > 0$ for all $r \in (0, r_{\mathrm{max}})$.

\begin{theoremEnd}[%
    % restate,
    proof here,
    text proof={},
    % category=ACmB2positivedefinite,
    % text link={See the \hyperref[proof:prAtEnd\pratendcountercurrent]{\textit{full proof}} in, p.~\pageref{proof:prAtEnd\pratendcountercurrent}.}
]{lemma}\label{thm:ACmB2_positive_definite}
    For any $\zeta, \gamma > 0$ and all $r \in (0, r_{\mathrm{max}})$, we have 
    $
        c^-\left(r, \zeta\right) > 0.
    $
\end{theoremEnd}
\begin{proofEnd}
    Consider the equivalence relations
    \begin{align*}
        & &
        c^-\left(r, \zeta\right) = p_1\left(r\right) - \sqrt{p_2\left(r\right) p_3\left(r\right)} &> 0
        % \\
        % &\Leftrightarrow&
        % \sqrt{p_2\left(r\right) p_3\left(r\right)} &> p_1\left(r\right)
        % \\
        % &\Leftrightarrow&
        % p_2\left(r\right) p_3\left(r\right) &> {p_1\left(r\right)}^2
        \\
        &\Leftrightarrow&
        {p_1\left(r\right)}^2- p_2\left(r\right) p_3\left(r\right) &> 0
        % \\
        % &\Leftrightarrow&
        % \left(
        %     a_1 \left( -2 e_1 - \left(e_2 + e_3\right) \right) 
        %     + 
        %     b_1^2 - b_2 b_3
        % \right) r^2 
        % + 
        % \left( 2 b_1 e_1 - \left(b_2 e_3 + b_3 e_2 \right) \right) r
        % & >
        % 0
        \\
        &\Leftrightarrow&
        \left(
            - \left(\frac{4}{3} - \delta^2\right) \zeta^2 - 2 \delta \left(1 - \delta\right) \zeta - 3 {\left(1 - \delta\right)}^2
        \right) \, r^2
        + 
        2 \zeta \left(1 - \delta^2\right) \, r
        &>
        0 \; .
    \end{align*}
    The last equivalence was derived using symbolic computation.
    Since $\delta \in (0, 1)$ and $\zeta > 0$, all of the coefficients on the left-hand side of the last line are non-zero.
    Also, it is a quadratic function with two roots: $r = 0$ and $r = r_{\mathrm{max}}$, and its leading coefficient is negative, implying that it is a concave quadratic. 
    Therefore, it is strictly positive in between the two roots, which turns out to be the open interval $(0, r_{\mathrm{max}})$.
    The equivalence with the condition $c^- > 0$ implies the result.
\end{proofEnd}

The remaining condition $\chi_A > 0$ automatically follows by choosing $c = c^-\left(r, \zeta\right)$.

\begin{theoremEnd}[%
    % restate,
    proof here,
    text proof={},
    % category=Apositivedefinite,
    % text link={See the \hyperref[proof:prAtEnd\pratendcountercurrent]{\textit{full proof}} in, p.~\pageref{proof:prAtEnd\pratendcountercurrent}.}
]{lemma}\label{thm:A_positive_definite}
    For any $\zeta, \gamma, r > 0$, if $c^-\left(r, \zeta\right) > 0$ then $\chi_{A}\left(r, \zeta, \gamma, c^-\left(r, \zeta\right)\right) > 0$.
\end{theoremEnd}
\begin{proofEnd}
First, notice that
\[
    \chi_A\left(r, \zeta, \gamma, c\right) > 0 
    \quad\Leftrightarrow\quad
    \left( - \zeta^2 - 3{\left(1 - \delta\right)}^2 \right) r^2 + 2 \zeta r > c \; .
\]
Denote the left-hand side as a polynomial in $r$ as
\[
    \left( -\zeta^2 - 3 {\left(1 - \delta\right)}^2 \right) r^2 + 2 \zeta r
    \quad\triangleq\quad p_4\left(r\right) 
\]
and pick $c = c^-\left(r, \zeta\right)$.
Then
\begin{align*}
    & &
    0 \;\;&<\;\; \chi_A\left(r, \zeta, \gamma, c^-\left(r, \zeta\right)\right)
    \\
    &\Leftarrow&
    p_1\left(r\right) - \sqrt{p_2\left(r\right) p_3\left(r\right)} \;\;&<\;\;  p_4\left(r\right)
    % \\
    % &\Leftrightarrow&
    % p_1\left(r\right) - p_4\left(r\right)  \;\;&<\;\; \sqrt{p_2\left(r\right) p_3\left(r\right)}
    \\
    &\Leftarrow&
    {\left( p_1\left(r\right) - p_4\left(r\right) \right)}^2 \;\;&<\;\; p_2\left(r\right) p_3\left(r\right) 
    \\
    &\Leftrightarrow&
    {\left( p_1\left(r\right) - p_4\left(r\right) \right)}^2 - p_2\left(r\right) p_3\left(r\right) \;\;&<\;\;0 
\end{align*}
The left-hand side forms a polynomial, which follows from symbolic computation, 
\[
    {\left( p_1\left(r\right) - p_4\left(r\right) \right)}^2 - p_2\left(r\right) p_3\left(r\right)
    \quad=\quad
    p_5\left(r\right) r^2
    \quad=\quad
    \left( a_5 r^2 + b_5 r + e_5 \right) r^2
\]
with the coefficients 
\begin{align*}
    a_5 
    \quad&\triangleq\quad
    3 \delta \left(2 - \delta\right)\left( {\left(1 - \delta\right)}^2 + 1 \right) - \frac{1}{3} \zeta^4 - 2 \zeta^2 {\left(1 - \delta\right)}^2 - 3 
    \\
    b_5 
    \quad&\triangleq\quad
    -6 \delta {\left(1 - \delta\right)}^3 + \frac{2}{3}\zeta^3 - 2 \zeta^2 \delta \left(1 - \delta\right) + 2 \zeta {\left(1 - \delta\right)}^2
    \\
    e_5 
    \quad&\triangleq\quad
    -3 \delta^2 {\left(1 - \delta\right)}^2 - \frac{1}{3}\zeta^2 + 2 \zeta \delta \left(1 - \delta\right) 
    \; .
\end{align*}

Analyzing the coefficient $a_5$ reveals that $p_5$ is a concave quadratic.
Specifically, $a_5 \rvert_{\zeta = 0} = 0$, and $a_5$ is monotonically decreasing with respect to $\zeta$ since, for all $\zeta > 0$, 
\begin{align*}
    \frac{\mathrm{d} a_5}{\mathrm{d} \zeta}
    &= 
    \frac{\mathrm{d}}{\mathrm{d} \zeta}
    \left\{
        3 \delta 
        \left(
        1 + \left(1 - \delta\right) + {\left(1 - \delta\right)}^2 + {\left(1 - \delta\right)}^3
        \right)
        - \frac{1}{3} \zeta^4 - 2 \zeta^2 {\left(1 - \delta\right)}^2 - 3
    \right\}
    % \\
    % &= 
    % 3 \delta 
    % \left(
    % - 1 
    % - \left(1 - \delta\right) - {\left(1 - \delta\right)}^2 - {\left(1 - \delta\right)}^3
    % \right)
    % +
    % 3 \delta 
    % \left(
    % \delta + 2 \delta {\left(1 - \delta\right)} + 3 \delta {\left(1 - \delta\right)}^2
    % \right)
    % \\
    % &\qquad
    % - 
    % \frac{4}{3} \zeta^3 
    % - 4 \zeta {\left(1 - \delta\right)}^2
    % - 4 \zeta^2 {\left(1 - \delta\right)} \delta
    % \\
    % &= 
    % 3 \delta 
    % \left(
    % - 2 \left(1 - \delta\right) \left(1 - \delta\right) 
    % - \left(1 - 3 \delta\right) {\left(1 - \delta\right)}^2 - {\left(1 - \delta\right)}^3
    % \right)
    % \\
    % &\qquad
    % - 
    % \frac{4}{3} \zeta^3 
    % - 4 \zeta {\left(1 - \delta\right)}^2
    % - 4 \zeta^2 {\left(1 - \delta\right)} \delta
    \\
    &= 
    - 12 \delta {\left(1 - \delta\right)}^3 -  \frac{4}{3} \zeta^3 
    - 4 \zeta {\left(1 - \delta\right)}^2
    - 4 \zeta^2 {\left(1 - \delta\right)} \delta
    \quad<\quad 0
    \; .
\end{align*}
That is, $a_5 < 0$.
Furthermore, using symbolic computation, the discriminant of $p_5$ follows as $\operatorname{disc}\left(p_5\right) = b_5^2 - 4 a_5 e_5 = 0$, meaning that $p_5$ has a unique root.
Said differently, for all $r \in \mathbb{R}$, $p_{5}\left(r\right) < 0$, meaning that, for any $\zeta > 0$ and all $r > 0$, $c^-\left(r, \zeta\right) = p_1\left(r\right) - \sqrt{p_2(r) p_3(r)} < p_4\left(r\right)$ holds.
\end{proofEnd}

Therefore, ensuring that $\chi_{AC-B^2}$ is non-negative also results in $\chi_A$ being positive.
As such, we have characterized the conditions under which $\chi_{A}$ and $\chi_{AB-C^2}$ are positive.
We are now ready to formally prove \cref{thm:general_contraction}.

\printProofs[generalcontraction]

% The bulk of the proof is dedicated to analyzing \(\chi_{AC-B^2}\) in \cref{section:proof_ACmB2_positive_definite}, which yields the following result:
% \input{thm_ACmB2_positive_definite}

% \(\chi_A\), on the other hand, is straightforward given \cref{eq:cstar}.

% \input{thm_A_positive_definite}

% \printProofs[generalcontraction]

\newpage
\subsubsection{Proof of \cref{thm:ACmB2_distinct_roots}}\label{section:proof_ACmB2_distinct_roots}
\printProofs[ACmB2distinctroots]

% \newpage
% \subsubsection{Proof of \cref{thm:ACmB2_positive_definite}}\label{section:proof_ACmB2_positive_definite}
% \printProofs[ACmB2positivedefinite]

% \newpage
% \subsubsection{Proof of \cref{thm:A_positive_definite}}\label{section:proof_A_positive_definite}
% \printProofs[Apositivedefinite]

\newpage
\subsubsection{Proof of \cref{thm:special_case_convergence_bound}}\label{section:proof_special_case_convergence_bound}

Recall $c^-\left(r, \zeta\right)$ in \cref{eq:cminus}.
Define
\begin{equation}
    r_{\mathrm{lin}} \quad\triangleq\quad \frac{ \zeta {(1 - \delta)} }{ 2 \zeta^2 + 6 {(1 - \delta)}^2 } \; .
    \label{eq:rlin}
\end{equation}
We establish that, for any $\zeta > 0$ and all $r \in (0, r_{\mathrm{lin}}]$, the inequality $c^-\left(r, \zeta\right) \geq \zeta r$ holds.
This is equivalent to establishing, for all $r \in (0, r_{\mathrm{lin}}]$, the inequality
\begin{align}
    p_1\left(r\right) - \sqrt{ p_2\left(r\right) p_3\left(r\right) } \quad\geq\quad \zeta r \; .
    \label{eq:linear_approximation_valid_condition}
\end{align}
To proceed, we will use the following fact:

\begin{theoremEnd}[%
    % restate,
    proof here,
    % text proof={},
    % category=p1largerthanzetar,
    % text link={See the \hyperref[proof:prAtEnd\pratendcountercurrent]{\textit{full proof}} in, p.~\pageref{proof:prAtEnd\pratendcountercurrent}.}
]{lemma}\label{thm:p1_larger_than_zeta_r}
    For all $r \in (0, r_{\mathrm{lin}}]$, $p_1\left(r\right) \geq \zeta r$ holds.
\end{theoremEnd}
\begin{proofEnd}
Consider the equivalence relations
\begin{align}
    & &
    p_1\left(r\right) \quad&>\quad \zeta r
    \nonumber
    \\
    &\Leftrightarrow&
    p_1\left(r\right) -  \zeta r \quad&>\quad 0
    \nonumber
    \\
    &\Leftrightarrow&
	- a_1 r^2 + \left(b_1 - \zeta \right) r + e_1
    \quad&>\quad 0 
    \label{eq:linear_approximation_condition}
    \; .
\end{align}
The left-hand side forms a concave quadratic.
Thus, if we find two points where this quadratic is strictly positive, all points in between satisfy \cref{eq:linear_approximation_condition}.
Since $e_1 > 0$, $r = 0$ trivially satisfies \cref{eq:linear_approximation_condition}.
On the other hand, since
\begin{align}
    r_{\mathrm{lin}} 
    \quad<\quad
    \frac{\zeta \left(1 - \delta\right)}{ \frac{2}{3} \zeta^2 + 2 {\left(1 - \delta\right)}^2 }
    \quad=\quad
    \frac{\zeta \left(1 - \delta\right)}{ a_1 } \; ,
    \label{eq:rlin_upper_bound}
\end{align}
for $r = r_{\mathrm{lin}}$, we have
\begin{align*}
    &-a_1 r^2 + \left(b_1 - \zeta\right) r_{\mathrm{lin}} + e_1 
    \\
    &\quad>
    - a_1 \frac{\zeta \left(1 - \delta\right)}{a_1} r_{\mathrm{lin}} - \delta \left(1 - \delta\right) r_{\mathrm{lin}} + \frac{1}{2} \left(1 - \delta^2\right)
    &&\text{(\cref{eq:rlin_upper_bound})}
    \\
    &\quad=
    -\left( \zeta  + \delta \right) \left(1 - \delta\right) r_{\mathrm{lin}} + \frac{1}{2} \left(1 - \delta^2\right) 
    \\
    &\quad=
    \left(1 - \delta\right)
    \left\{
    -\left( \zeta  + \delta \right) r_{\mathrm{lin}} + \frac{1}{2} \left(1 + \delta\right) 
    \right\}
    \\
    &\quad=
    \left(1 - \delta\right)
    \left\{
    -\left( \zeta  + \delta \right) \frac{\zeta \left(1 - \delta\right)}{2 \zeta^2 + 6 {\left(1 - \delta\right)}^2 } + \frac{1}{2} \left(1 + \delta\right) 
    \right\}
    \\
    &\quad=
    \left(1 - \delta\right)
    \left\{
    -\frac{\zeta^2 \left(1 - \delta\right)}{2 \zeta^2 + 6 {\left(1 - \delta\right)}^2  }
    + \frac{\delta \zeta \left(1 - \delta\right)}{2 \zeta^2 + 6 {\left(1 - \delta\right)}^2 }
    + \frac{1}{2} \left(1 + \delta\right) 
    \right\}
    \\
    &\quad>
    \left(1 - \delta\right)
    \left\{
    -\frac{\zeta^2 \left(1 - \delta\right)}{2 \zeta^2} 
    + \frac{\delta \zeta \left(1 - \delta\right)}{6 {\left(1 - \delta\right)}^2 }
    + \frac{1}{2} \left(1 + \delta\right) 
    \right\}
    &&\text{($\zeta > 0$, $(1 - \delta) > 0$)}
    \\
    &\quad=
    \left(1 - \delta\right)
    \left\{
    - \frac{1}{2} \left(1 - \delta\right)
    + \frac{1}{6} \frac{\delta  \zeta}{1 - \delta}
    + \frac{1}{2} \left(1 + \delta\right) 
    \right\}
    \\
    &\quad=
    \left(1 - \delta\right)
    \left( \delta + \frac{1}{6} \frac{\delta \zeta}{1 - \delta} \right)
    \\
    &\quad> 0
    \; .
    &&\text{($0 < \delta < 1$)}
\end{align*}
Therefore, all points $r \in (0, r_{\mathrm{lin}}]$ satisfy $p_1(r) > \zeta r$.
\end{proofEnd}

Then, for all $r \in (0, r_{\mathrm{lin}}]$, \cref{eq:linear_approximation_valid_condition} can be developed as follows:
\begin{align}
    & &
    p_1\left(r\right) - \zeta r \quad&\geq\quad \sqrt{ p_2\left(r\right) p_3\left(r\right) }
    \nonumber
    \\
    &\Leftrightarrow&
    {\left( p_1\left(r\right) - \zeta r \right)}^2 \quad&\geq\quad p_2\left(r\right) p_3\left(r\right) 
    &&\text{($p_1(r) > \zeta r$)}
    \nonumber
    \\
    &\Leftrightarrow&
    {\left( p_1\left(r\right) - \zeta r \right)}^2 - p_2\left(r\right) p_3\left(r\right) \quad&\geq\quad 0 \; .
    \label{eq:linar_approximation_condition}
\end{align}
Denote the left-hand side, which follows from symbolic computation, as
\begin{align}
    {\left( p_1\left(r\right) - \zeta r \right)}^2 - p_2\left(r\right) p_3\left(r\right) 
    \quad\triangleq\quad
    p_6\left(r\right) r
    \quad=\quad
    \left( a_6 r^{2} - b_6 r  + e_6  \right) r
    \nonumber
\end{align}
with the coefficients 
\begin{align*}
    a_6 
    \quad&=\quad
    \frac{4}{3} \zeta^3 + 4 \zeta {\left(1 - \delta\right)}^2
    \\
    b_6 
    \quad&=\quad  
    3 {\left( 1 - \delta \right)}^2 + \zeta^2 \left(\frac{7}{3} - \delta^2\right)
    \\
    e_6 \quad&=\quad \zeta \left(1 - \delta^{2}\right)
    \; .
\end{align*}
Under the conditions on the parameters, $p_6$ can be shown to be non-negative.

\begin{theoremEnd}[%
    proof here,
    % restate,
    % text proof={},
    % category=p6nonnegative,
    % text link={See the \hyperref[proof:prAtEnd\pratendcountercurrent]{\textit{full proof}} in, p.~\pageref{proof:prAtEnd\pratendcountercurrent}.}
]{lemma}\label{thm:p6_non_negative}
    For all $r \in (0, r_{\mathrm{lin}}]$, $p_6\left(r\right) \geq 0$ holds.
\end{theoremEnd}
\begin{proofEnd}
Since $p_6$ is a convex quadratic, it suffices to verify that
\begin{equation}
    \frac{\mathrm{d}p_6}{\mathrm{d}r}\left(r_{\mathrm{lin}}\right) \quad\leq\quad 0 \qquad\text{and}\qquad p_6\left(r_{\mathrm{lin}}\right) \quad\geq\quad 0 \; .
    \label{eq:p6_positive_condition}
\end{equation}
Then, by the monotonicity of the derivative of convex functions, all $r \in (0, r_{\mathrm{lin}})$ satisfy $p_6\left(r\right) > 0$.
First, notice that
\begin{align}
    r_{\mathrm{lin}}
    &=
    \frac{2}{3}
    \frac{ \zeta^2 {\left(1 - \delta\right)} }{ (4/3) \zeta^3 + 4 \zeta {\left(1 - \delta\right)}^2 }
    \nonumber
    \\
    &= 
    \frac{2}{3}
    \frac{\zeta^2 {(1 - \delta)} }{a_6}
    \label{eq:rlin_p6_relationship}
    \\
    &= 
    \frac{2}{7}
    \frac{\zeta^2 { \left(\frac{7}{3} - \frac{7}{3} \delta\right)} }{a_6}
    \nonumber
    \\
    &<
    \frac{2}{7}
    \frac{\zeta^2 { \left(\frac{7}{3} - \delta^2\right)} }{a_6}
    &&\text{($ - (7/3) \delta < -\delta < -\delta^2$)}
    \nonumber
    \\
    &<
    \frac{2}{7}
    \frac{ 3 {\left(1 - \delta\right)}^2 +  \zeta^2 \left(\frac{7}{3} - \delta^2\right) }{a_6}
    &&\text{($0 < 3 {\left(1 - \delta\right)}^2$)}
    \nonumber
    \\
    &=
    \frac{2}{7}
    \frac{b_6}{a_6}
    \; .
    \label{eq:rlin_b6_a6}
\end{align}
Then \cref{eq:rlin_b6_a6} implies
\begin{align*}
    \frac{\mathrm{d}p_6}{\mathrm{d}r}\left(r_{\mathrm{lin}}\right)
    \quad=\quad
    2 a_6 r_{\mathrm{lin}} - b_6
    \quad<\quad
    - \frac{3}{7} b_6
    \quad<\quad
    0 \; .
\end{align*}
Furthermore, using \cref{eq:rlin_p6_relationship} and symbolic computation,
\begin{align*}
    p_6\left(r_{\mathrm{lin}}\right)
    \quad=\quad
	a_6 r_{\mathrm{lin}}^2 - b_6 r_{\mathrm{lin}} + e_6
    \quad=\quad
    \frac{ \frac{4}{9} \zeta^4 {\left(1 - \delta\right)}^2 - \frac{2}{3} \zeta^2 {\left(1 - \delta\right)} b_6 + a_6 e_6 }{a_6}
    \; .
\end{align*}
The denominator $a_6$ satisfies $a_6 > 0$.
Therefore, we only need to analyze the sign of the numerator.
Denote the numerator as a function of $\zeta$, which follows from symbolic computation, as
\[
    \frac{4}{9} \zeta^4 {\left(1 - \delta\right)}^2 - \frac{2}{3} \zeta^2 {\left(1 - \delta\right)} b_6 + a_6 e_6
    \quad\triangleq\quad
    p_7\left(\zeta\right) \zeta^2
    \quad=\quad
    a_7 \zeta^4 + b_7  \zeta^2
\]
with the coefficients
\begin{align*}
    a_7 \quad&\triangleq\quad 
    \frac{2}{9} \left(1 + 3 \delta\right)\left(1 - \delta\right) \left(1 + \delta\right)
    \\
    b_7 \quad&\triangleq\quad 
    2 {\left(1 - \delta\right)}^3 \left(1 + 2 \delta \right)
    \; .
\end{align*}
By inspection, it is clear that $a_7, b_7 > 0$ for all $\delta \in (0, 1)$.
Therefore, $p_7 > 0$ for all $\zeta > 0$.
We can then conclude that $p_6\left(r\right) > 0$ for all $r \in (0, r_{\mathrm{lin}}]$.
\end{proofEnd}

Thus, the inequality in \cref{eq:linar_approximation_condition} holds.
We are now ready to prove \cref{thm:special_case_convergence_bound}.

\printProofs[specialcaseconvergencebound]

% \newpage
% \subsubsection{Proof of \cref{thm:p1_larger_than_zeta_r}}\label{section:proof_p1_larger_than_zeta_r}
% \printProofs[p1largerthanzetar]

% \newpage
% \subsubsection{Proof of \cref{thm:p6_non_negative}}\label{section:proof_p6_non_negative}
% \printProofs[p6nonnegative]

\newpage
\subsection{Asymptotic Bias Analysis}\label{section:proof_asymptotic_bias}

For the analysis of the asymptotic bias, we will define the following:
\begin{itemize}
    \item ${(Z_k)}_{k \geq 0}$ is a Markov chain following the kinetic Langevin diffusion discretized with the stochastic exponential integrator (\cref{eq:kinetic_langevin_monte_carlo_algorithm}), where, for some arbitrary distribution $\mu \in \mathcal{P}(\mathbb{R}^{2 d})$, it is initialized as $Z_{0} \sim \mu$.
    
    \item ${(Z_t^*)}_{t \geq 0}$ is the kinetic Langevin dynamics (\cref{eq:kinetic_langevin_dynamics}) initialized from its stationary distribution $\pi$.

    \item ${(Z_t')}_{t \in [hk, h(k+1)]}$ is, for each $k \geq 0$, the stochastic exponential integration of ${(Z^*_t)}_{t \geq 0}$ zeroth-order interpolated over the interval $[hk, h(k+1)]$. Specifically, for any $k \geq 0$ and any $t \in [hk, h(k+1)]$,
    \begin{align}
        V_{t}^{\prime}
        &=
        \mathrm{e}^{- \gamma (t - hk)} V_{hk}^* - \eta \int^t_{hk} \mathrm{e}^{- \gamma \left(s - hk\right)} \nabla U\left(X_{hk}^{*}\right) \mathrm{d}s + \sqrt{2 \gamma \eta} \int^t_{hk} \mathrm{e}^{-\gamma \left(s - hk\right)} \, \mathrm{d}B_s
        \nonumber
        \\
        X_t^{\prime} &= X_{hk}^* + \int^t_{hk} V_s^{\prime} \, \mathrm{d}s \; .
        \label{eq:interpolated_process}
    \end{align}
    In essence, ${(Z_t')}_{t \in [hk, h (k+1)]}$ is a kinetic Langevin diffusion process with the drift set to be the zeroth order interpolation of the drift of ${(Z_{hk}^*)}_{k \geq 0}$.
\end{itemize}
Throughout the proof, we will assume that for any $k \geq 0$ and any $t \in [hk, h(k+1)]$, ${(Z_k)}_{k \geq 0}$, ${(Z_t^*)}_{t \geq 0}$, and ${(Z_t')}_{t \geq 0}$ are synchronously coupled by sharing the same noise process ${(B_t)}_{t \geq 0}$.

\subsubsection{Proof of \cref{thm:asymptotic_bias}}
Under \cref{assumption:condition_linear_approximation_contraction}, \cref{thm:special_case_convergence_bound} asserts that ${(Z_k)}_{k \geq 0}$ converges to the unique stationary distribution of $K$, $\pi_h$, such that $\lim_{k \to \infty} \mu K^k = \lim_{k \to \infty} \mathrm{Law}\left(Z_k\right) = \pi_h$.
Also, since ${(Z^*_t)}$ is initialized from its stationary distribution $\pi$, for all $t \geq 0$, its law is $\mathrm{Law}\left(Z^*_t\right) = \pi$.
Notice 
\begin{align}
    \mathrm{W}_{a,b}(\pi_h, \pi)
    \quad=\quad
    \lim_{k \to \infty}
    \mathrm{W}_{a,b}(\mu K^k, \pi P_{hk})
    \quad=\quad
    \lim_{k \to \infty}
    \mathrm{W}_{a,b}(\mathrm{Law}(Z_k), \mathrm{Law}(Z_{hk}^*)) \; .
\end{align}
Since $\mathrm{W}_{a,b}$ is a proper distance metric under $a = 4/\gamma^2$ and $b = 1/\gamma$.
Then, for any $k \geq 0$, we have the decomposition
\begin{align}
    &
    \mathrm{W}_{a,b}(\mathrm{Law}(Z_{k+1}), \mathrm{Law}(Z_{h (k+1)}^*))
    \nonumber
    \\
    &\qquad\leq
    \mathrm{W}_{a,b}(\mathrm{Law}(Z_{k+1}), \mathrm{Law}(Z_{h (k+1)}'))
    \;+\;
    \mathrm{W}_{a,b}(\mathrm{Law}(Z'_{h(k+1)}), \mathrm{Law}(Z_{h (k+1)}^*)) \; .
    \nonumber
\end{align}
Furthermore, \cref{thm:special_case_convergence_bound} asserts that a contraction holds as
\begin{align}
    &\mathrm{W}_{a,b}(\mathrm{Law}(Z_{k+1}), \mathrm{Law}(Z_{h (k+1)}^*))
    \nonumber
    \\
    &\quad\leq
    {(1 - \tilde{c})}^{1/2}
    \mathrm{W}_{a,b}(\mathrm{Law}(Z_{k}), \mathrm{Law}(Z_{h k}'))
    \;+\;
    \underbrace{
    \mathrm{W}_{a,b}(\mathrm{Law}(Z'_{h(k+1)}), \mathrm{Law}(Z_{h (k+1)}^*)) 
    }_{\triangleq E_{\mathrm{disc}}}
    \label{eq:bias_partial_contraction_general}
\end{align}
By bounding the one-step local discretization error $E_{\mathrm{disc}}$ and iterating the recursion, we obtain a bound on the asymptotic bias.

The following lemma relates the total local error to the local error of the momentum alone:

\begin{theoremEnd}[%
    category=generaldiscretizationbound,
    text link={\noindent\textit{Proof.} The proof is deferred to \cref{section:proof_general_discretization_bound}. \qed},
    text proof = {Proof of \string\pratendRef{thm:prAtEnd\pratendcountercurrent}}
]{lemma}\label{thm:general_discretization_bound}
    Suppose $b^2 \leq  a$ holds.
    Then, for any $k \geq 0$,
    \[
        \mathrm{W}_{a,b}\left(\mathrm{Law}\left(Z'_{hk}\right), \mathrm{Law}\left(Z_{h k}^*\right)\right) 
        \quad\leq\quad
        \widetilde{E}_{\mathrm{pos}}
        +
        \widetilde{E}_{\mathrm{mom}}
        \; ,
    \]
    where
    \begin{align*}
        \widetilde{E}_{\mathrm{pos}}
        \;\triangleq\;
        {\left\{
        h
        \int_{h k}^{h (k + 1)} 
        \mathbb{E} \norm{ 
            V_{t}^* - 
            V^{\prime}_{t} 
        }^2 
        \, \mathrm{d}t 
        \right\}}^{1/2}
        \quad\text{and}\quad
        \widetilde{E}_{\mathrm{mom}}
        \;\triangleq\;
        {\left\{
            a \, \mathbb{E} \norm{  V_{h(k+1)}^* - V^{\prime}_{h (k+1)}  }^2
        \right\}}^{1/2} 
        \; .
    \end{align*}
\end{theoremEnd}
\begin{proofEnd}
For each $k \geq 0$, ${(Z_t^*)}$ over the time interval $[hk, h(k+1)]$, 
\begin{align}
    &\mathrm{W}_{a,b}\left(\mathrm{Law}\left(Z'_{h (k+1)}\right), \mathrm{Law}\left(Z_{h (k+1)}^*\right)\right)
    \nonumber
    \\
    &\qquad\leq
    \mathbb{E}^{1/2} {\lVert Z_{h(k+1)}^* - Z^{\prime}_{h (k+1)} \rVert}_{a,b}^2 \; .
    \nonumber
\shortintertext{Applying \cref{thm:minkowski_inequality},} 
    &\qquad\leq
    \mathbb{E}^{1/2} {\lVert X_{h(k+1)}^* - X^{\prime}_{h (k+1)} \rVert}^2
    +
    \sqrt{a} \,  \mathbb{E}^{1/2} {\lVert V_{h(k+1)}^* - V^{\prime}_{h (k+1)} \rVert}^2 
    \nonumber
    \\
    &\qquad=
    \mathbb{E}^{1/2} \norm{ \int_{h k}^{h (k + 1)} \left( V_{t}^* - V^{\prime}_{t}\right) \mathrm{d}t + \left(X_{hk}^* - X_{hk}'\right)
    }^2
    +
    \sqrt{a} \, \mathbb{E}^{1/2} \norm{  V_{h(k+1)}^* - V^{\prime}_{h (k+1)}  }^2
    \nonumber
\shortintertext{and the fact that ${(Z_{t}')}_{t \geq 0}$ and ${(Z_{t}^*)}_{t \geq 0}$ are synchronously coupled,} 
    &\qquad=
    \mathbb{E}^{1/2} \norm{ \int_{h k}^{h (k + 1)} \left( V_{t}^* - V^{\prime}_{t}\right) \mathrm{d}t 
    }^2
    +
    \sqrt{a} \,
    \mathbb{E}^{1/2} \norm{  V_{h(k+1)}^* - V^{\prime}_{h (k+1)}  }_2^2 \; .
    \nonumber
\shortintertext{The result follows by applying Jensen's inequality.}
    &\qquad\leq
    {\Bigg(
    h
    \int_{h k}^{h (k + 1)} 
    \mathbb{E} \norm{ 
        V_{t}^* - 
        V^{\prime}_{t} 
    }^2 
    \, \mathrm{d}t 
    \Bigg)}^{1/2}
    +
    \sqrt{a} \, 
    \mathbb{E}^{1/2} \norm{  V_{h(k+1)}^* - V^{\prime}_{h (k+1)}  }^2
    \; .
\end{align}
\end{proofEnd}

Here, $\widetilde{E}_{\mathrm{pos}}$ is the local error of the position ${(X_k)}_{k \geq 0}$, while $\widetilde{E}_{\mathrm{mom}}$ is that of the momentum ${(V_k)}_{k \geq 0}$.
The second-order behavior of KLMC can be seen by the fact that $\widetilde{E}_{\mathrm{pos}}$ has an extra integral with a factor of $h$.

Using the fact that ${(Z_t')}_{t \geq 0}$ is the zeroth-order interpolation of ${(Z_t^*)}_{t \geq 0}$ and that ${(Z_t^*)}_{t \geq 0}$ is stationary, the local error of the momentum can be bounded as follows:
\begin{theoremEnd}[%
    category=momentumdeviationbound,
    text proof = {Proof of \string\pratendRef{thm:prAtEnd\pratendcountercurrent}},
    text link={\noindent{\textit{Proof.}} The proof is deferred to \cref{section:proof_momentum_deviation_bound}. \qed}
]{lemma}\label{thm:momentum_deviation_bound}
Suppose, \cref{asusmption:hessian_bounded} holds.
Then, for any $k \geq 0$,
\begin{align*}
    \mathbb{E} {\lVert V_{h (k+1)}^* - V_{h (k+1)}^{\prime} \rVert}^2
    &\leq
    \frac{1}{4}
    d \beta^2 \eta^3
    \left\{
    \frac{h}{\gamma^3} 
    -
    \mathrm{e}^{-2 h \gamma }
    \left(
    \frac{h}{\gamma^3} 
    +
    2
    \frac{h^2}{\gamma^2}  
    +
    2
    \frac{h^3}{\gamma}
    \right)
    \right\} 
    \\
    \int_{h k}^{h (k + 1)} 
    \mathbb{E} \norm{ 
        V_{t}^* - 
        V^{\prime}_{t} 
    }^2 
    \, \mathrm{d}t 
    &\leq
    \frac{1}{8}
    d \beta^2 \eta^3
    \Bigg\{
    \frac{h^2}{\gamma^3}
    -
    3
    \frac{1}{\gamma^5}
    +
    \mathrm{e}^{- 2 h \gamma}
    \left(
    3
    \frac{1}{\gamma^5}
    +
    6
    \frac{h}{\gamma^4}
    +
    5
    \frac{h^2}{\gamma^3}
    +
    2
    \frac{h^3}{\gamma^2} 
    \right) 
    \Bigg\} \; .
\end{align*}
\end{theoremEnd}
\begin{proofEnd}
Since \({(V_t^*)}_{t \geq 0}\) and \({(V_t^{\prime})}_{t \geq 0}\) are synchronously coupled and $V^*_{h k} = V_{hk}'$ for all $t \in [hk, h(k+1)]$,
\begin{align}
    &
    \mathbb{E} \norm{ V_t^* - V_{t}^{\prime} }^2
    \nonumber
    \\
    &\;\;=
    \mathbb{E} \norm{
    V_t^* - V_{h k}^*
    +
    V_{h k}^{\prime}
    -
    V_{t}^{\prime}
    }^2
    &&\text{($V_{hk}^{\prime} = V_{hk}^{*}$)}
    \nonumber
    \\
    &\;\;=
    \mathbb{E} \left\lVert
    \mathrm{e}^{-\gamma (t - hk)} V_{hk}^*
    +
    \eta
    \int_{hk}^{t} 
    \mathrm{e}^{- \gamma (s - hk)}
    \nabla U\left(X_s^*\right)
    \, \mathrm{d} s
    \right.
    \nonumber
    \\
    & &
    \mathllap{
    -
    \left.
    \mathrm{e}^{-\gamma (t - h k)} V_{hk}^*
    +
    \eta
    \int_{hk}^{t} 
    \mathrm{e}^{- \gamma (s - hk)}
    \nabla U\left(X_{hk}^*\right)
    \, \mathrm{d} s
    \right\rVert^2
    }
    \nonumber
    \\
    &\;\;=
    \mathbb{E} \norm{
    \eta
    \int_{hk}^{t} 
    \mathrm{e}^{- \gamma (s - hk)}
    \left(
    \nabla U\left(X_s^*\right)
    -
    \nabla U\left(X_{hk}^{*}\right)
    \right)
    \, \mathrm{d} s
    }^2
    \nonumber
    \\
    &\;\;\leq
    \eta^2
    (t - hk)
    \int_{hk}^{t} 
    \mathrm{e}^{- 2 \gamma (s - hk)}
    \mathbb{E}
    {\lVert
    \nabla U\left(X_{s}^*\right)
    -
    \nabla U(X_{hk}^*)
    \rVert}^2
    \, \mathrm{d} s 
    \nonumber
    &&\text{(Jensen's inequality)}
    \\
    &\;\;\leq
    \beta^2
    \eta^2
    (t - hk)
    \int_{hk}^{t} 
    \mathrm{e}^{- 2 \gamma (s - hk)}
    \mathbb{E}
    \norm{ X_{s}^* - X_{hk}^* }^2
    \, \mathrm{d} s 
    \nonumber
    &&\text{(\cref{asusmption:hessian_bounded})}
    \\
    &\;\;=
    \beta^2
    \eta^2
    (t - hk)
    \int_{hk}^{t} 
    \mathrm{e}^{- 2 \gamma (s - hk)}
    \mathbb{E}
    \norm{ \int^s_{hk} V_u^* \mathrm{d}u }^2
    \, \mathrm{d} s 
    \nonumber
    \\
    &\;\;\leq
    \beta^2
    \eta^2
    (t - hk)
    \int_{hk}^{t} 
    \mathrm{e}^{- 2 \gamma (s - hk)}
    (s - hk)
    \int^{s}_{hk} 
    \mathbb{E}
    {\lVert V_u^* \rVert}^2 \, \mathrm{d} u
    \, \mathrm{d} s
    \; .
    \nonumber
    &&\text{(Jensen's inequality)}
\end{align}
Since ${(V^*_t)}_{t \geq 0}$ is stationary with the stationary distribution \(\mathrm{N}\left(0_d, \eta \mathrm{I}_d\right)\), 
\begin{align}
    \mathbb{E} \norm{ V_t^* - V_{t}^{\prime} }^2
    \;\;&\leq\;\;
    d \beta^2 \eta^3
    (t - hk)
    \int_{hk}^{t} 
    \mathrm{e}^{- 2 \gamma (s - hk)}
    (s - hk)
    \int^{s}_{hk} \mathrm{d} u \, \mathrm{d} s
    \nonumber
    \\
    \;\;&=\;\;
    d \beta^2 \eta^3 \left(t - hk\right)
    \int_{hk}^{t} 
    \mathrm{e}^{- 2 \gamma (s - hk)}
    {(s - hk)}^2
    \, \mathrm{d} s 
    \; .
    \label{eq:momentum_deviation_bound_raw}
\end{align}

Using the function $I_p(t)$, \cref{eq:momentum_deviation_bound_raw} becomes
\begin{align}
    &\mathbb{E} \norm{ V_t^* - V_{t}^{\prime} }^2
    \nonumber
    \\
    &\leq
    d \beta^2 \eta^3 \left(t - hk\right)
    I_2\left(t\right)
    \nonumber
    \\
    &=
    d \beta^2 \eta^3
    \left(t - hk\right)
    \left\{
    \frac{1}{4 \gamma^3} 
    -
    \frac{1}{4 \gamma^3} 
    \mathrm{e}^{-2 \gamma (t - hk)}
    - 
    \frac{ {\left(t - hk\right)} }{2 \gamma^2}  \mathrm{e}^{-2 \gamma \left(t - hk\right)}
    - 
    \frac{{\left(t - hk\right)}^2 }{2 \gamma} \mathrm{e}^{-2 \gamma \left(t - hk\right)}
    \right\} 
    \nonumber
    \\
    &=
    d \beta^2 \eta^3
    \left\{
    \frac{\left(t - hk\right)}{4 \gamma^3} 
    -
    \frac{\left(t - hk\right)}{4 \gamma^3} 
    \mathrm{e}^{-2 \gamma (t - hk)}
    - 
    \frac{ {\left(t - hk\right)}^2 }{2 \gamma^2}  \mathrm{e}^{-2 \gamma \left(t - hk\right)}
    - 
    \frac{{\left(t - hk\right)}^3 }{2 \gamma} \mathrm{e}^{-2 \gamma \left(t - hk\right)}
    \right\} 
    \; .
    \label{eq:momentum_deviation_bound_first}
\end{align}
The first inequality in the statement follows by substituting $t = h (k+1)$.

The second inequality follows by integrating \cref{eq:momentum_deviation_bound_first} again.
Then
\begin{align}
    &\int^t_{h k}
    \mathbb{E} {\lVert V_s^* - V_{s}^{\prime} \rVert}^2 \, \mathrm{d}s
    \nonumber
    \\
    &\;\leq
    \nonumber
    d \beta^2 \eta^3
    \Bigg\{
        \frac{ {(t - hk)}^2 }{8\gamma^3}
        - \frac{1}{4 \gamma^3} I_1\left(t\right)
        - \frac{1}{2 \gamma^2} I_2\left(t\right)
        - \frac{1}{2 \gamma} I_3\left(t\right)
    \Bigg\}
    \nonumber
    \\
    &\;=
    d \beta^2 \eta^3
    \Bigg\{
    \frac{{\left( t - hk \right)}^2 }{8\gamma^3}
    -
    \frac{3}{8}
    \frac{1}{\gamma^5}
    \nonumber
    \\
    &\qquad\qquad\qquad
    +\mathrm{e}^{- 2 \gamma \left(t - hk\right) }
    \left(
    \frac{3}{8}
    \frac{1}{\gamma^5}
    +
    \frac{3}{4}
    \frac{1}{\gamma^4}
    {\left( t - hk \right)}
    +
    \frac{5}{8}
    \frac{1}{\gamma^3}
    {\left( t - hk \right)}^2
    +
    \frac{1}{4} 
    \frac{1}{\gamma^2} 
    {\left( t - hk \right)}^3
    \right) 
    \Bigg\} \; .
    \nonumber
\end{align}
Substituting $t = h (k + 1)$ yields the statement.

\end{proofEnd}

Using these supporting results, we are now ready to prove \cref{thm:asymptotic_bias}.

\printProofs[asymptoticbias]

\newpage
\subsubsection{Proof of \cref{thm:general_discretization_bound}}\label{section:proof_general_discretization_bound}

For the proof of \cref{thm:general_discretization_bound}, we will use the following Minkowski-type inequality that holds for the norm $\norm{\cdot}_{a,b}$:

\begin{lemma}\label{thm:minkowski_inequality}
    Suppose $b^2 \leq a$ holds.
    Then, for any pair of random variables $Z = (X,V)$, the norm $\norm{\cdot}_{a,b}$ satisfies 
    \begin{align*}
        \mathbb{E}^{1/2}\norm{ Z }_{a,b}^2
        \quad\leq\quad
        \mathbb{E}^{1/2}\norm{X}^2 + \sqrt{a} \mathbb{E}^{1/2}\norm{V}^2 \; .
    \end{align*}
\end{lemma}
\begin{proof}
\begin{align*}
    \mathbb{E}^{1/2}\norm{ Z }_{a,b}^2
    \;\;&=\;\;
    \sqrt{
        \mathbb{E}\norm{ X }^2 + 2 b \mathbb{E}  \inner{X}{V} + a \mathbb{E} \norm{V}^2
    }
    \\
    \;\;&\leq\;\;
    \sqrt{
        \mathbb{E}\norm{ X }^2 + 2 b {\left( \mathbb{E} \norm{X}^2 \right)}^{1/2} {\left( \norm{V}^2 \right)}^{1/2} + a \mathbb{E} \norm{V}^2
    }
    &&\text{(Cauchy--Schwarz)}
    \\
    \;\;&\leq\;\;
    \sqrt{
        \mathbb{E}\norm{X}^2 + 2 \, {( \mathbb{E} \norm{X}^2)}^{1/2} \sqrt{a}{\left( \, \mathbb{E} \norm{V}^2 \right)}^{1/2}  + a \mathbb{E} \norm{V}^2
    }
    &&\text{($b^2 \leq a$)}
    \\
    \;\;&=\;\;
    \mathbb{E}^{1/2}\norm{X}^2 + \sqrt{a} \mathbb{E}^{1/2}\norm{V}^2 \; .
\end{align*}
\end{proof}

\printProofs[generaldiscretizationbound]

\newpage
\subsubsection{Proof of \cref{thm:momentum_deviation_bound}}\label{section:proof_momentum_deviation_bound}

In the proof, we will use the function
\begin{align*}
    I_p\left(t\right) \quad\triangleq\quad \int^t_{hk} {\left(s - hk\right)}^p \mathrm{e}^{-2 \gamma \left(s - hk\right)} \, \mathrm{d}s \; .
\end{align*}
For any \(p \geq 2\), integration by part yields the recursive definition
\begin{align*}
    I_1\left(t\right) \quad&= \quad\frac{1}{4\gamma^2} - \frac{1}{4\gamma^2} \mathrm{e}^{-2 \gamma \left(t - hk\right)} - \frac{1}{2 \gamma} \mathrm{e}^{-2 \gamma \left(t - hk\right)} \left(t - hk\right) 
    \\
    I_p\left(t\right) \quad&=\quad - \frac{1}{2 \gamma} {\left( t - kh \right)}^p \mathrm{e}^{- 2\gamma \left(t - hk\right)} + \frac{p}{2 \gamma} I_{p-1}\left(t\right)
    &&\text{(for $p > 1$)}
    \; .
\end{align*}
The exact form of \(I_p\) for any order \(p > 1\) can then be conveniently computed via symbolic computation.

\printProofs[momentumdeviationbound]

\newpage
\subsubsection{Proof of \cref{thm:asymptotic_bias_special_cases}}\label{section:proof_asymptotic_bias_special_cases}

Recall $E_{\mathrm{pos}}$ and $E_{\mathrm{mom}}$ in \cref{thm:asymptotic_bias}.
The result follows by upper-bounding the terms depending on $\zeta$, $f_{\mathrm{pos}}$ for $E_{\mathrm{pos}}$ and $f_{\mathrm{mom}}$ for $E_{\mathrm{mom}}$ as
\begin{align*}
    f_{\mathrm{pos}}\left(\zeta\right) 
    \quad&\triangleq \quad
    \zeta^2 - 3 + \exp\left(-2\zeta\right) \left(3 + 6\zeta + 5\zeta^2 + 2\zeta^3\right) 
    \\
    f_{\mathrm{mom}}\left(\zeta\right) 
    \quad&\triangleq\quad
    1 - \exp\left(-2\zeta\right) \left( 1 + 2 \zeta + 2 \zeta^2 \right)
    \; .
\end{align*}
The upper bound for the underdamped and overdamped regimes corresponds to upper bounding $f_{\mathrm{pos}}$ and $f_{\mathrm{mom}}$ with their asymptotes in the direction of $\zeta \to 0$ and $\zeta \to \infty$, respectively.

For obtaining the upper bound, we will use the fact that, for two differentiable functions $f, g : \mathbb{R} \to \mathbb{R}$, suppose $f' \leq g'$ and there exists some $\zeta_0 \in \mathbb{R}$ such that $f(\zeta_0) \leq g(\zeta_0)$.
Then the fundamental theorem of calculus yields $f \leq g$.
This strategy can be applied recursively such that, for any $n \geq 2$, if the $n$th derivative satisfies $f^{(n)} \leq g^{(n)}$ and there exists a collection of points ${(\zeta_0^{(k)})}_{k = 1, \ldots, n - 1}$ such that, for all $k = 1, \ldots, n - 1$, the bounds $f^{(k)}(\zeta_0^{(k}) \leq g^{(k)}(\zeta_0^{(k})$ hold, then $f \leq g$.
Therefore, analyzing the derivatives of $f_{\mathrm{pos}}$ and $f_{\mathrm{mom}}$ will yield our bounds.

The derivatives of $f_{\mathrm{pos}}$ follow as
\begin{align}
    \frac{\mathrm{d}f_{\mathrm{pos}}}{\mathrm{d}\zeta}\left(\zeta\right)
    \quad&= \quad
    2 \zeta \left( 1 - \mathrm{e}^{-2 \zeta} \left(\zeta^2 + {\left( 1 + \zeta\right)}^2\right) \right)
    \label{eq:fpos_derivative1}
    \\
    \frac{\mathrm{d}^2f_{\mathrm{pos}}}{\mathrm{d}\zeta^2}\left(\zeta\right)
    \quad&=\quad
    2 \mathrm{e}^{-2 \zeta} \left( 8\zeta^3 - 4\zeta^2 - 4\zeta\right) + 2 - 4 \mathrm{e}^{-2 \zeta}
    \label{eq:fpos_derivative2}
    \\
    \frac{\mathrm{d}^3 f_{\mathrm{pos}}}{\mathrm{d}\zeta^3}\left(\zeta\right)
    \quad&=\quad
    -16 \, \mathrm{e}^{-2 \zeta} \left(\zeta - 2\right) \zeta^2 \; ,
    \label{eq:fpos_derivative3}
\end{align}
while the derivative of $f_{\mathrm{mom}}$ follows as
\begin{align}
    \frac{\mathrm{d} f_{\mathrm{mom}}}{ \mathrm{d}\zeta }\left(\zeta\right) 
    \quad=\quad
    4 \mathrm{e}^{-2\zeta}\zeta^2 \; .
    \label{eq:fmom_derivative1}
\end{align}

\printProofs[asymptoticbiasspecialcases]

\newpage
\subsection{Complexity Analysis (Proof of \cref{thm:complexity})}\label{section:proof_complexity}

From our choice of $a$ and $b$, $\norm{\cdot}_{a,b}$ is a valid norm meaning that we have 
\[
    \mathrm{W}_{a,b}\left(\mu K^n, \pi\right)
    \quad\leq\quad
    \mathrm{W}_{a,b}\left(\mu K^n, \pi_h\right)
    +
    \mathrm{W}_{a,b}\left(\pi_h, \pi\right) \; .
\]
Thus, for any given $\epsilon > 0$, the result follows by solving for the smallest $n \geq 0$ ensuring 
\begin{equation}
    \mathrm{W}_{a,b}\left(\mu K^n, \pi_h\right)
    +
    \mathrm{W}_{a,b}\left(\pi_h, \pi\right)
    \;\leq\; \epsilon
    \quad\Rightarrow\quad
    \mathrm{W}_{a,b}\left(\mu K^n, \pi\right) \;\leq\; \epsilon \; .
    \label{eq:complexity_error_condition1}
\end{equation}
Denote the non-stationarity error as $E_{\mathrm{stat}} \triangleq \mathrm{W}_{a,b}\left(\mu K^n, \pi_h\right)$ and recall $E_{\mathrm{mom}}$ and $E_{\mathrm{pos}}$ in \cref{thm:asymptotic_bias}.
Then the condition in \cref{eq:complexity_error_condition1} can be ensured by
\begin{align}
    E_{\mathrm{stat}} \leq \frac{\epsilon}{3} \quad\land\quad
    E_{\mathrm{pos}} \leq \frac{\epsilon}{3} \quad\land\quad
    E_{\mathrm{mom}} \leq \frac{\epsilon}{3} \; .
    \label{eq:complexity_error_condition3}
\end{align}
$E_{\mathrm{stat}}$ is bounded by \cref{thm:special_case_convergence_bound}, while $E_{\mathrm{stat}}$ and $E_{\mathrm{pos}}$ is bounded by \cref{thm:asymptotic_bias_special_cases}.

% \begin{lemma}\label{thm:parameter_condition_sufficient_choice}
%     If the parameters satisfy $\zeta \leq 1$, $\gamma = \sqrt{7 L}$, $\eta = 1$, then \cref{eq:condition_linear_approximation_contraction} holds.
% \end{lemma}
% \begin{proof}
%     Consider the bound   
%     \[
%         \frac{\zeta}{1 - \mathrm{e}^{-\zeta}} \quad\leq\quad 1 + \frac{\zeta}{2} + \frac{\zeta^2}{6} \; .
%     \]
%     Then \cref{eq:condition_linear_approximation_contraction} follows from 
%     \begin{align}
%         \frac{\eta}{\gamma^2} \left( 2 \frac{\zeta}{{\left(1 - \delta\right)}} + 6 \right) 
%         \quad\leq\quad  \frac{1}{L} 
%         \quad&\Leftarrow\quad
%         \frac{\eta}{\gamma^2} \left( 2 \left(1 + \frac{\zeta}{2} + \frac{\zeta^2}{6}\right) + 6 \right) 
%         \quad\leq\quad \frac{1}{L} 
%         \nonumber
%         \\
%         \quad&\Leftrightarrow\quad
%         8 + \zeta + \frac{\zeta^2}{3}
%         \quad\leq\quad 
%         \frac{\gamma^2}{L \eta} 
%         \nonumber
%         \\
%         \quad&\Leftarrow\quad
%         {\left(\zeta + 2\right)}^2 + 20
%         \quad\leq\quad \frac{3 \gamma^2}{L \eta} 
%         \nonumber
%         \\
%         \quad&\Leftarrow\quad
%         \zeta
%         \quad\leq\quad \sqrt{ \frac{3 \gamma^2}{L \eta} - 20 } \; .
%         \nonumber
%     \end{align}
%     Under the choice of parameters in the statement, the right-hand side is equal to 1.
%     Therefore, it suffices to ensure $\zeta \leq 1$.
% \end{proof}

Let's first solve for the number of steps that guarantees $E_{\mathrm{pos}}$ is small.

\begin{lemma}\label{thm:Estat_condition}
    Suppose \cref{asusmption:hessian_bounded} and \cref{assumption:condition_linear_approximation_contraction} hold, and $n$ denotes the number of KLMC steps.
    Then, for any $\epsilon > 0$,
    \[
        n
        \quad\geq\quad 
        \frac{2 \gamma^2}{\zeta \eta \alpha}  
        \log\left( {\mathrm{W}_{a,b}\left(\mu, \pi_h\right)} \frac{1}{\epsilon} \right) 
        \qquad\Rightarrow\qquad
        E_{\mathrm{stat}} \quad\leq\quad \epsilon \; .
    \]
\end{lemma}
\begin{proof}
    The result follows as a corollary of \cref{thm:special_case_convergence_bound}, which implies
    \begin{align}
        E_{\mathrm{stat}} 
        \quad\leq\quad
        {\left(1 - \frac{\zeta \eta \alpha}{\gamma^2}\right)}^{n/2} {\mathrm{W}_{a,b}\left(\mu, \pi_h\right)}
        \quad\leq \quad
        \exp\left(-\frac{\zeta \eta \alpha}{2 \gamma^2} n\right) {\mathrm{W}_{a,b}\left(\mu, \pi_h\right)} \; .
    \end{align}
    This yields the necessary number of iterations through the implications
    \begin{align}
        E_{\mathrm{stat}} \quad\leq\quad \epsilon
        \quad&\Leftarrow\quad
        &\exp\left(-\frac{\zeta \alpha}{2 \gamma^2} n\right) {\mathrm{W}_{a,b}\left(\mu, \pi_h\right)} 
        \quad&\leq\quad 
        \epsilon
        \nonumber
        \\
        \quad&\Leftrightarrow\quad
        &\exp\left(\frac{\zeta \alpha}{2 \gamma^2} k\right)
        \quad&\geq\quad {\mathrm{W}_{a,b}\left(\mu, \pi_h\right)} \frac{1}{\epsilon} 
        \nonumber
        \\
        \quad&\Leftrightarrow\quad
        &\frac{\zeta \alpha}{2 \gamma^2} n
        \quad&\geq\quad \log\left( {\mathrm{W}_{a,b}\left(\mu, \pi_h\right)} \frac{1}{\epsilon} \right)
        \nonumber
        \\
        \quad&\Leftrightarrow\quad
        &n
        \quad&\geq\quad 
        \frac{2 \gamma^2}{\zeta \alpha}  
        \log\left( {\mathrm{W}_{a,b}\left(\mu, \pi_h\right)} \frac{1}{\epsilon} \right)
        \nonumber
        \; .
    \end{align}
\end{proof}

The condition on the discretization error can be ensured by making $\zeta$ small enough.
\begin{lemma}\label{thm:Emom_Epos_condition}
    Suppose \cref{asusmption:hessian_bounded} and \cref{assumption:condition_linear_approximation_contraction} hold.
    Then, for any $\epsilon > 0$,
    \begin{align*}
        \zeta
        \leq
        \min\left\{  \;
            \sqrt{\frac{15}{4}} \frac{\epsilon^{1/2} \gamma^{1/2}}{d^{1/4} \kappa^{1/2} \eta^{1/4}}, \;
            \frac{\sqrt{3}}{4} \frac{\epsilon \gamma}{d^{1/2} \kappa \eta^{1/2}}
        \; \right\}
        \quad&\Rightarrow\quad
        E_{\mathrm{pos}} \leq \epsilon
        \;\;\land\;\;
        E_{\mathrm{mom}} \leq \epsilon
        \; .
    \end{align*}
\end{lemma}
\begin{proof}
    Under the stated conditions, we can invoke \cref{thm:asymptotic_bias_special_cases}, which implies
    \begin{align*}
        E_{\mathrm{pos}}
        \quad\leq\quad
        \frac{4}{15} \frac{d^{1/2} \kappa \eta^{1/2}}{\gamma} \zeta^2
        \qquad\text{and}\qquad
        E_{\mathrm{mom}}
        \quad\leq\quad
        \frac{4}{\sqrt{3}} 
        \frac{d^{1/2} \kappa \eta^{1/2}}{\gamma} \zeta \; .
    \end{align*}
    We can solve for the conditions $E_{\mathrm{pos}} \leq \epsilon$ and $E_{\mathrm{mom}} \leq \epsilon$:
    \begin{align*}
        E_{\mathrm{pos}} \quad\leq\quad \epsilon
        \qquad&\Leftrightarrow\qquad
        &\frac{4}{15} \frac{d^{1/2} \kappa \eta^{1/2}}{\gamma} \zeta^2
        \quad&\leq\quad
        \epsilon
        \\
        \qquad&\Leftrightarrow\qquad
        &
        \frac{15}{4}
        \frac{\epsilon \gamma}{d^{1/2} \kappa \eta^{1/2}} 
        \quad&\geq\quad
        \zeta^2
        \\
        \qquad&\Leftrightarrow\qquad
        &
        \sqrt{\frac{15}{4}}
        \frac{\epsilon^{1/2} \gamma^{1/2}}{d^{1/4} \kappa^{1/2} \eta^{1/4}} 
        \quad&\geq\quad
        \zeta
        \\
        E_{\mathrm{mom}} \quad\leq\quad \epsilon
        \qquad&\Leftrightarrow\qquad
        &\frac{4}{\sqrt{3}} 
        \frac{d^{1/2} \kappa \eta^{1/2}}{\gamma} \zeta
        \quad&\leq\quad
        \epsilon
        \\
        \quad&\Leftrightarrow\quad
        &
        \frac{\sqrt{3}}{4} 
        \frac{\epsilon\gamma}{d^{1/2} \kappa \eta^{1/2}} 
        \quad&\geq\quad
        \zeta \; .
    \end{align*}
    Taking the minimum over the two upper bounds on $\zeta$ ensures that both $E_{\mathrm{pos}} \leq \epsilon$ and $E_{\mathrm{mom}} \leq \epsilon$ are satisfied simultaneously.
\end{proof}

\printProofs[complexity]

\newpage
\paragraph{Acknowledgements}
The authors thank Saeed Saremi for helpful discussions, the reviewers who suggested various improvements, and the organizers of the FDS Conference: Recent Advances and Future Directions for Sampling hosted in 2024 at Yale University, New Haven, U.S., for enabling this project.

Part of the work was done while K. Kim was interning at Genentech, South San Francisco, U.S.
K. Kim was supported through the NSF award [IIS2145644].

\paragraph{Supplementary Material}
The code for replicating our symbolic computation results is available online\footnote{
URL: \url{https://gist.github.com/Red-Portal/d157fd23cf254b61dd36379920abd7d3}
}, along the instructions for executing the code\footnote{
URL: \url{https://gist.github.com/Red-Portal/98835096ca82203231e046bffc980ebf}
}.

%\bibliographystyle{alpha}
%\bibliography{references}
\printbibliography

% \newpage
% \appendix 
% \input{section_exponential_integrator}

\end{document}